\makeatletter
\def\input@path{{styles/}{../styles/}}
\makeatother

\ifx\DCG\undefined%
    \newcommand{\DCGVer}[1]{}%
    \newcommand{\RegVer}[1]{#1}%
    \documentclass[11pt]{article}%
    \usepackage[cm]{fullpage}
    \def\UseBibLatex{1}
\else
    \newcommand{\DCGVer}[1]{#1}%
    \newcommand{\RegVer}[1]{}%
    \documentclass[pdflatex,sn-mathphys]{sn-jnl}%
\fi

\DCGVer{%
   \usepackage{natbib}%
   \usepackage{anyfontsize}
}

\ifx\sarielComp\undefined%
\newcommand{\SarielComp}[1]{}
\newcommand{\NotSarielComp}[1]{#1}%
\else
\newcommand{\SarielComp}[1]{#1}%
\newcommand{\NotSarielComp}[1]{}%
\fi

\usepackage{amsmath}%
\usepackage{amssymb}%
\usepackage{xcolor}%
\usepackage{graphicx}
\usepackage{animate}%
\RegVer{\usepackage{mathabx}}%
\usepackage{bm}%
\usepackage[clock]{ifsym}%

\SarielComp{\usepackage{sariel_colors}}%

\RegVer{%

   \usepackage[amsmath,thmmarks]{ntheorem}%
   \theoremseparator{.}%

   \usepackage{titlesec}%
   \titlelabel{\thetitle. }%

}

\DCGVer{%
   \RequirePackage{stmaryrd,textcomp}%
}

\usepackage{xcolor}%
\usepackage{mleftright}%
\usepackage{xspace}%
\usepackage{hyperref}%
\usepackage{caption}%

\RegVer{%
   \definecolor{DarkBlue}{rgb}{0,0,0.50}
}
\DCGVer{%
   \definecolor[named]{DarkBlue}{HTML}{3D7391}%
}
\usepackage{hyperref}%
\hypersetup{%
      breaklinks,%
      colorlinks=true,%
      urlcolor=DarkBlue,%
      linkcolor=[rgb]{0.5,0.0,0.0},%
      citecolor=[rgb]{0,0.2,0.445},%
      filecolor=[rgb]{0,0,0.4},
      anchorcolor=[rgb]={0.0,0.1,0.2}%
}
\usepackage[ocgcolorlinks]{ocgx2}

\providecommand{\TPDF}[2]{\texorpdfstring{#1}{#2}}
\usepackage{anyfontsize}

\newcommand{\UsePackage}[1]{%
  \IfFileExists{../styles/#1.sty}{%
      \usepackage{../styles/#1}%
   }{%
      \IfFileExists{./styles/#1.sty}{%
         \usepackage{styles/#1}%
      }{%
         \usepackage{#1}%
      }%
   }%
}

\providecommand{\BibLatexMode}[1]{}
\providecommand{\BibTexMode}[1]{#1}

\ifx\UseBibLatex\undefined%
  \renewcommand{\BibLatexMode}[1]{}
  \renewcommand{\BibTexMode}[1]{#1}
\else
  \renewcommand{\BibLatexMode}[1]{#1}
  \renewcommand{\BibTexMode}[1]{}
\fi

\BibLatexMode{%
   \UsePackage{sariel_biblatex}%
   \usepackage[english]{babel}%
   \usepackage{csquotes}
   \renewcommand*{\multicitedelim}{\addcomma\space}%
}

\DCGVer{%
\theoremstyle{thmstyleone}%
\newtheorem{theorem}{Theorem}%
\newtheorem{proposition}[theorem]{Proposition}%
\newtheorem{lemma}[theorem]{Lemma}%
\newtheorem{corollary}[theorem]{Corollary}%

\theoremstyle{thmstyletwo}%
\newtheorem{example}{Example}%
\newtheorem{remark}{Remark}%

\theoremstyle{thmstylethree}%
\newtheorem{definition}{Definition}%

\raggedbottom

}

\RegVer{%
\theoremstyle{plain}%
\newtheorem{theorem}{Theorem}[section]

\newtheorem{lemma}[theorem]{Lemma}

\newtheorem{corollary}[theorem]{Corollary}

\theoremstyle{plain}%
\theoremheaderfont{\sf}%
\theorembodyfont{\upshape}%
\newtheorem*{remark:unnumbered}[theorem]{Remark}%
\newtheorem{remark}[theorem]{Remark}%
\newtheorem{definition}[theorem]{Definition}

\newcommand{\myqedsymbol}{\rule{2mm}{2mm}}

\theoremheaderfont{\em}%
\theorembodyfont{\upshape}%
\theoremstyle{nonumberplain}%
\theoremseparator{}%
\theoremsymbol{\myqedsymbol}%
\newtheorem{proof}{Proof:}%

\newtheorem{proof:sketch}{Proof (sketch):}%

}

\definecolor{blue25emph}{rgb}{0, 0, 11}
\providecommand{\emphic}[2]{%
   \textcolor{blue25emph}{%
      \textbf{\emph{#1}}}%
   \index{#2}}

\providecommand{\emphi}[1]{\emphic{#1}{#1}}

\definecolor{almostblack}{rgb}{0, 0, 0.3}
\providecommand{\emphw}[1]{{\textcolor{almostblack}{\emph{#1}}}}%

\newcommand{\atgen}{\symbol{'100}}
\newcommand{\SarielThanks}[1]{\thanks{Department of Computer Science;
      University of Illinois; 201 N. Goodwin Avenue; Urbana, IL,
      61801, USA; {\tt sariel\atgen{}illinois.edu};%
      {\tt\url{https://sarielhp.org}}. #1}}

\newcommand{\HLink}[2]{\hyperref[#2]{#1~\ref*{#2}}}
\newcommand{\HLinkSuffix}[3]{\hyperref[#2]{#1\ref*{#2}{#3}}}

\newcommand{\figlab}[1]{\label{fig:#1}}
\newcommand{\figref}[1]{\HLink{Figure}{fig:#1}}

\newcommand{\seclab}[1]{{\label{sec:#1}}}
\newcommand{\secref}[1]{\HLink{Section}{sec:#1}}

\newcommand{\thmlab}[1]{{\label{theo:#1}}}
\newcommand{\thmref}[1]{\HLink{Theorem}{theo:#1}}

\newcommand{\corlab}[1]{\label{cor:#1}}
\newcommand{\corref}[1]{\HLink{Corollary}{cor:#1}}%

\providecommand{\deflab}[1]{\label{def:#1}}
\newcommand{\defref}[1]{\HLink{Definition}{def:#1}}

\newcommand{\steplab}[1]{\label{step:#1}}
\newcommand{\stepref}[1]{\HLinkSuffix{}{step:#1}{}}

\newcommand{\lemlab}[1]{\label{lemma:#1}}
\newcommand{\lemref}[1]{\HLink{Lemma}{lemma:#1}}%

\newcommand{\tbllab}[1]{\label{table:#1}}
\newcommand{\tblref}[1]{\HLink{Table}{table:#1}}

\providecommand{\eqlab}[1]{}%
\renewcommand{\eqlab}[1]{\label{equation:#1}}
\newcommand{\Eqref}[1]{\HLinkSuffix{Eq.~(}{equation:#1}{)}}

\providecommand{\remove}[1]{}%
\newcommand{\Set}[2]{\left\{ #1 \;\middle\vert\; #2 \right\}}

\newcommand{\floor}[1]{\left\lfloor {#1} \right\rfloor}

\renewcommand{\th}{th\xspace}

\renewcommand{\Re}{\mathbb{R}}%
\newcommand{\reals}{\Re}%

\usepackage[inline]{enumitem}

\newlist{compactenumA}{enumerate}{5}%
\setlist[compactenumA]{topsep=0pt,itemsep=-1ex,partopsep=1ex,parsep=1ex,%
   label=(\Alph*)}%

\newlist{compactenumA*}{enumerate*}{5}%
\setlist[compactenumA*]{topsep=0pt,itemsep=-1ex,partopsep=1ex,parsep=1ex,%
   label=(\Alph*)}%

\newlist{compactenuma}{enumerate}{5}%
\setlist[compactenuma]{topsep=0pt,itemsep=-1ex,partopsep=1ex,parsep=1ex,%
   label=(\alph*)}%

\newlist{compactenumI}{enumerate}{5}%
\setlist[compactenumI]{topsep=0pt,itemsep=-1ex,partopsep=1ex,parsep=1ex,%
   label=(\Roman*)}%

\newlist{compactenumi}{enumerate}{5}%
\setlist[compactenumi]{topsep=0pt,itemsep=-1ex,partopsep=1ex,parsep=1ex,%
   label=(\roman*)}%

\newlist{compactitem}{itemize}{5}%
\setlist[compactitem]{topsep=0pt,itemsep=-1ex,partopsep=1ex,parsep=1ex,%
   label=\bullet}%

\DeclareFontFamily{U}{BOONDOX-calo}{\skewchar\font=45 }
\DeclareFontShape{U}{BOONDOX-calo}{m}{n}{
  <-> s*[1.05] BOONDOX-r-calo}{}
\DeclareFontShape{U}{BOONDOX-calo}{b}{n}{
  <-> s*[1.05] BOONDOX-b-calo}{}
\DeclareMathAlphabet{\mathcalb}{U}{BOONDOX-calo}{m}{n}
\SetMathAlphabet{\mathcalb}{bold}{U}{BOONDOX-calo}{b}{n}
\DeclareMathAlphabet{\mathbcalb}{U}{BOONDOX-calo}{b}{n}

\numberwithin{figure}{section}%
\numberwithin{table}{section}%
\numberwithin{equation}{section}%

\newcommand{\SaveContent}[2]{%
   \expandafter\newcommand{#1}{#2}%
}

\DeclareFontFamily{U}{BOONDOX-calo}{\skewchar\font=45 }
\DeclareFontShape{U}{BOONDOX-calo}{m}{n}{
  <-> s*[1.05] BOONDOX-r-calo}{}
\DeclareFontShape{U}{BOONDOX-calo}{b}{n}{
  <-> s*[1.05] BOONDOX-b-calo}{}
\DeclareMathAlphabet{\mathcalb}{U}{BOONDOX-calo}{m}{n}
\SetMathAlphabet{\mathcalb}{bold}{U}{BOONDOX-calo}{b}{n}
\DeclareMathAlphabet{\mathbcalb}{U}{BOONDOX-calo}{b}{n}

\newcommand{\UZ}[1]{\Mh{\Uparrow}_{#1}}%

\newcommand{\LX}[1]{\Mh{\Downarrow}\pth{#1}}%
\newcommand{\LY}[2]{\Mh{\Downarrow}_{#1}\pth{#2}}%

\newcommand{\eY}[2]{\Mh{\Updownarrow}_{#1}\pth{#2}}%

\newcommand{\uX}[1]{\Mh{\uparrow}\pth{#1}}%
\newcommand{\uY}[2]{\Mh{\uparrow}_{#1}\pth{#2}}%
\newcommand{\uZ}[1]{\Mh{\uparrow}_{#1}}%

\newcommand{\lX}[1]{\Mh{\downarrow}\pth{#1}}%
\newcommand{\lY}[2]{\Mh{\downarrow}_{#1}\pth{#2}}%
\newcommand{\lZ}[1]{\Mh{\downarrow}_{#1}}%

\newcommand{\opt}{\Mh{\mathsf{k}}}%
\newcommand{\optX}[1]{\Mh{\mathsf{k}}_{#1}}%
\newcommand{\optY}[2]{\Mh{\mathsf{k}}_{#1}(#2)}%
\newcommand{\optWX}[1]{\Mh{\mathsf{k}}^{\mathsf{w}}_{#1}}%
\newcommand{\optWY}[2]{\Mh{\mathsf{k}}^{\mathsf{w}}_{#1}(#2)}%
\newcommand{\optHX}[1]{\Mh{\mathsf{k}}^{\mathsf{h}}_{#1}}%

\newcommand{\etal}{\textit{et~al.}\xspace}
\newcommand{\etalNP}{\textit{et~al.}\xspace}
\newcommand{\sphere}{\ensuremath{\mathbb{S}}}
\newcommand{\HH}{\Mh{\mathcal{H}}}
\newcommand{\PP}{\Mh{\mathcal{P}}}
\newcommand{\SP}{\Mh{\mathcal{Z}}}

\newcommand{\ICS}{\Mh{\Xi}}%
\newcommand{\IS}{\Mh{\mathcal{I}}}%
\newcommand{\JS}{\Mh{\mathcal{J}}}%
\providecommand{\SS}{\Mh{S}}%
\renewcommand{\SS}{\Mh{S}}%
\newcommand{\LS}{\Mh{L}}%
\newcommand{\LSA}{\Mh{M}}%

\newcommand{\seg}{\Mh{s}}%

\newcommand{\intervalX}[1]{\Mh{\sphericalangle}{#1}}%

\newcommand{\OS}{\Mh{{K}}}%
\newcommand{\GS}{\Mh{G}}%
\newcommand{\QS}{\Mh{Q}}%
\newcommand{\segfn}{\Mh{\mathcalb{s}}}%
\newcommand{\segfnY}[2]{\segfn\pth{#1, #2}}
\newcommand{\shootfn}{\Mh{\mathcalb{f}}}%
\newcommand{\shoot}[2]{\shootfn\pth{#1,#2}}

\newcommand{\Line}{\Mh{\mathcalb{l}}}%
\newcommand{\Arr}{\Mh{\mathcal{A}}}%
\newcommand{\ArrX}[1]{\Arr\pth{#1}}%

\newcommand{\AIX}[1]{\sphericalangle{#1}}

\newcommand{\dualX}[1]{#1^{\star}}%
\newcommand{\polarX}[1]{#1^{\odot}}%

\newcommand{\dirs}{\Mh{\sphere}}

\newcommand{\sPntY}[2]{\overline{#1}_{#2}}

\newcommand{\hnSY}[2]{\hn_{#1,#2}}
\newcommand{\hpSY}[2]{\hp_{#1,#2}}
\newcommand{\LineSY}[2]{\Line_{#1,#2}}
\newcommand{\titleX}[1]{\textsc{#1}\xspace}

\newcommand{\CTY}[2]{#1 \sqcap #2}

\newcommand{\hp}{\overline{\hn}}%

\newcommand{\VX}[1]{\Mh{V}\pth{#1}}

\newcommand{\face}{\Mh{\mathsf{F}}}
\newcommand{\faceo}{\face_{\!\!\origin}}
\newcommand{\faceoX}[1]{\faceo\pth{#1}}

\usepackage{stmaryrd}%
\providecommand{\IntRange}[1]{\mleft\llbracket #1 \mright\rrbracket}
\newcommand{\IRX}[1]{\IntRange{#1}}%

\newcommand{\extX}[1]{\pa_{#1}}%

\newcommand{\extLX}[1]{\Line_{#1}}%

\newcommand{\hn}{\Mh{\mathcalb{h}}}%

\newcommand{\hnX}[1]{\hn\pth{#1}}%

\newcommand{\hpX}[1]{\hp\pth{#1}}%

\newcommand{\hnDirX}[1]{\hn_{#1}}%

\newcommand{\hpDirX}[1]{\hp_{#1}}%

\newcommand{\gnDirX}[1]{\Mh{\mathcalb{g}}_{#1}}%

\newcommand{\hnDirY}[2]{\hn_{#1}\pth{#2}}%

\newcommand{\hnDirMY}[2]{\hn_{#1} \ominus { #2 }}%

\newcommand{\hpDirMY}[2]{\hp_{#1} \ominus { #2 }}%

\providecommand{\Mh}[1]{#1}%

\newcommand{\PS}{\Mh{P}}%
\newcommand{\PSA}{\Mh{Q}}%
\newcommand{\CS}{\Mh{C}}%

\newcommand{\eps}{\varepsilon}%
\newcommand{\ProjWidthC}{\Mh{\overline{\omega}}}%

\newcommand{\pwY}[2]{\ProjWidthC\pth{#1, #2}}
\newcommand{\pwIY}[2]{\Mh{I}_{#1}\pth{ #2}}
\newcommand{\pwIX}[1]{\Mh{I}_{#1}}

\newcommand{\DotProd}[2]{\permut{{#1},{#2}}}
\newcommand{\permut}[1]{\left\langle {#1} \right\rangle}

\providecommand{\si}[1]{#1}

\newcommand{\InPoly}{\Mh{U}}%

\newcommand{\coreY}[2]{\Mh{\mathrm{core}}\pth{#1,#2}}%

\newcommand{\abs}[1]{{\lvert #1 \rvert}}%
\newcommand{\vabs}[1]{{\left\lvert #1 \right\rvert}}%
\newcommand{\cardin}[1]{\lvert #1 \rvert}%
\newcommand{\lenX}[1]{\lvert #1 \rvert}%
\newcommand{\slabY}[2]{\Mh{\mathrm{slab}}\pth{#1, #2}}%
\newcommand{\normX}[1]{\left\| {#1} \right\|}

\newcommand{\CH}{\Mh{\mathsf{ch}}}
\newcommand{\CHX}[1]{\CH\pth{#1}}
\newcommand{\CC}{\Mh{\mathcal{C}}}%

\newcommand{\interX}[1]{\mathrm{int}\pth{#1}}

\newcommand{\pa}{\Mh{p}}%
\newcommand{\pb}{\Mh{q}}%
\newcommand{\pc}{\Mh{s}}%
\newcommand{\pd}{\Mh{t}}%

\newcommand{\ellipse}{\Mh{\mathcal{E}}}%
\newcommand{\origin}{\Mh{\mathsf{o}}}%
\newcommand{\diskY}[2]{\Mh{\mathsf{d}}\pth{#1, #2}}%

\newcommand{\dY}[2]{\|#1-#2\|}

\newcommand{\cwX}[1]{\scalebox{0.8}{\Interval}\pth{#1}}%

\newcommand{\diamX}[1]{\mathrm{diam}\pth{#1}}%

\newcommand{\ndiag}{\Mh{\mathcal{N}}}%
\newcommand{\clockwiseX}[1]{\Mh{\mathrm{cw}}\pth{#1}}%

\usepackage{mleftright}
\newcommand{\pth}[2][\!]{\mleft({#2}\mright)}%
\newcommand{\tgamma}{\Mh{\xi}}%
\newcommand{\TG}{\Mh{\Xi}}%
\newcommand{\tI}{\Mh{J}}

\newcommand{\UX}[1]{\Mh{\Uparrow}\pth{#1}}%
\newcommand{\UXX}[1]{\Mh{\Uparrow}_{#1}}%
\newcommand{\UY}[2]{\UXX{#1}\pth{#2}}%

\BibLatexMode{%
   \bibliography{opt_kernel}%
}%

\DCGVer{%
   \jyear{2023}%
}
\begin{document}

\RegVer{%
   \title{Computing Instance-Optimal Kernels in Two Dimensions%
   }%

   \author{%
      Pankaj K. Agarwal\thanks{Department of Computer Science, Duke
         University, Durham NC 27708. Work on this paper was supported
         by NSF grants \si{IIS-18-14493} and CCF-20-07556.}  \and%
      Sariel Har-Peled\SarielThanks{Work on this paper was partially
         supported by a NSF AF award
         CCF-1907400.  %
   }%
}%
}%

\DCGVer{%
   \title{%
      Computing Instance-Optimal Kernels in Two
      Dimensions%
   }%

      \author[1]{\fnm{Pankaj} \sur{K. Agarwal}
      \email{pankaj@cs.duke.edu}}

   \author[2]{\fnm{Sariel} \sur{Har-Peled}
      \email{sariel@illinois.edu}}

   \affil[1]{%
      \orgdiv{Department of Computer Science}, %
      \orgname{Duke University}, %
      \orgaddress{%
         \street{Levine Science Research Center D315}, %
         \city{Durham}, %
         \postcode{27708}, %
         \state{NC}, %
         \country{USA}}%
   }%
   \affil[2]{%
      \orgdiv{Department of Computer Science}, %
      \orgname{University of Illinois Urbana-Champaign}, %
      \orgaddress{%
         \street{201 N. Goodwin Avenue}, %
         \city{Urbana}, %
         \postcode{61801}, %
         \state{IL}, %
         \country{USA}}%
   }%
}

\RegVer{\date{\today}}%

\RegVer{\maketitle}%

\RegVer{%
   \begin{abstract}
       Let $\PS$ be a set of $n$ points in $\reals^2$.  For a
       parameter $\eps\in (0,1)$, a subset $\CS\subseteq \PS$ is an
       \emph{$\eps$-kernel} of $\PS$ if the projection of the convex
       hull of $\CS$ approximates that of $\PS$ within
       $(1-\eps)$-factor in every direction. The set $\CS$ is a
       \emph{weak $\eps$-kernel} of $\PS$ if its directional width
       approximates that of $\PS$ in every direction.  Let
       $\optY{\eps}{\PS}$ (resp.\ $\optWY{\eps}{\PS}$) denote the
       minimum-size of an $\eps$-kernel (resp. weak $\eps$-kernel) of
       $\PS$. We present an $O(n\optY{\eps}{\PS}\log n)$-time
       algorithm for computing an $\eps$-kernel of $\PS$ of size
       $\optY{\eps}{\PS}$, and an $O(n^2\log n)$-time algorithm for
       computing a weak $\eps$-kernel of $\PS$ of size
       $\optWY{\eps}{\PS}$.  We also present a fast algorithm for the
       Hausdorff variant of this problem.

       In addition, we introduce the notion of \emph{$\eps$-core}, a
       convex polygon lying inside $\CHX{\PS}$, prove that it is a
       good approximation of the optimal $\eps$-kernel, present an
       efficient algorithm for computing it, and use it to compute an
       $\eps$-kernel of small size.
   \end{abstract}%
}

\DCGVer{%
   \abstract{ Let $\PS$ be a set of $n$ points in $\reals^2$.  For a
      parameter $\eps\in (0,1)$, a subset $\CS\subseteq \PS$ is an
      \emph{$\eps$-kernel} of $\PS$ if the projection of the convex
      hull of $\CS$ approximates that of $\PS$ within
      $(1-\eps)$-factor in every direction. The set $\CS$ is a
      \emph{weak $\eps$-kernel} of $\PS$ if its directional width
      approximates that of $\PS$ in every direction.  Let
      $\optY{\eps}{\PS}$ (resp.\ $\optWY{\eps}{\PS}$) denote the
      minimum-size of an $\eps$-kernel (resp. weak $\eps$-kernel) of
      $\PS$. We present an $O(n\optY{\eps}{\PS}\log n)$-time algorithm
      for computing an $\eps$-kernel of $\PS$ of size
      $\optY{\eps}{\PS}$, and an $O(n^2\log n)$-time algorithm for
      computing a weak $\eps$-kernel of $\PS$ of size
      $\optWY{\eps}{\PS}$.  We also present a fast algorithm for the
      Hausdorff variant of this problem.

      \smallskip%
      In addition, we introduce the notion of \emph{$\eps$-core}, a
      convex polygon lying inside $\CHX{\PS}$, prove that it is a good
      approximation of the optimal $\eps$-kernel, present an efficient
      algorithm for computing it, and use it to compute an
      $\eps$-kernel of small size.  %

      \bigskip A preliminary version of this paper to appear in SoCG
      2023.%

 } }%

\maketitle{}%

\section{Introduction}
\seclab{intro}

Coresets have been successfully used as geometric summaries to develop
fast approximation algorithms for a wide range of geometric
optimization problems.  Agarwal~\etal~\cite{ahv-aemp-04} introduced
the notions of $\eps$-kernels/coresets for approximating the convex
hull of a point set $\PS$ in $\reals^d$: For an interval $J=[a,b]$,
let
\begin{math}
    (1-\eps)J = [a+(\eps/2)\lenX{J},b-(\eps/2)\lenX{J}]
\end{math}
be its scaling down by a factor of $1-\eps$ around its center.  For a
direction $v\in\dirs$, let $\pwIY{v}{\PS}$ denote the projection of
$\CHX{\PS}$ in direction $v$, which is an interval.  A subset
$\CS \subseteq \PS$ is an $\eps$-kernel if
$\pwIY{v}{\CS}\supseteq (1-\eps)\pwIY{v}{\PS}$ for all directions
$v\in\dirs$, see \defref{approx:and:kernel}.  The \emph{weak
   $\eps$-kernels} impose a weaker requirement that
$\lenX{\pwIY{v}{\CS}} \ge (1-\eps)\lenX{\pwIY{v}{\PS}}$ for all
$v\in\dirs$, see \defref{weak:kernel} and \figref{weak:strong}.

\begin{figure}[ht]
    \phantom{}
    \hfill%
    \includegraphics[page=2]{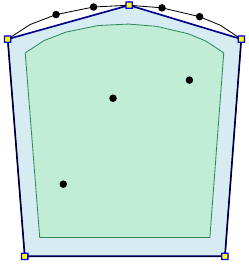}%
    \hfill%
    \includegraphics[page=3]{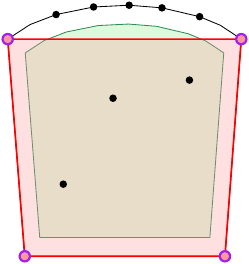}
    \hfill%
    \phantom{}
    \caption{Somewhat oversimplifying the difference, a regular kernel
       has to conceptually include a ``shrunken'' middle portion
       (left), while the weak kernel (right) only has to approximate
       the projections. Specifically, on the left, the projection
       interval of the approximation has to include the projection
       interval of the green region. On the right,
       the approximation projection interval needs to be sufficiently long
       but it does not have the inclusion constraint.}

    \figlab{weak:strong}
\end{figure}

It is known that there exists an $\eps$-kernel (as well as a weak
$\eps$-kernel) of $\PS$ of size $O(\eps^{-(d-1)/2})$ and that it can
be computed efficiently~\cite{ahv-aemp-04}.  However there may exist
an $\eps$-kernel of $\PS$ of much smaller size, as is often the case
in practice, see, e.g.~\cite{yapv-pmsfk-08}. Let $\optY{\eps}{\PS}$ be
the minimum size of an $\eps$-kernel of $\PS$.  An interesting
question is whether an $\eps$-kernel of $\PS$ of size $\optX{\eps}$
can be computed efficiently, i.e., computing an
\emph{instance-optimal} $\eps$-kernel.  A similar question can be
asked for weak $\eps$-kernels.  These problems are known to be NP-Hard
for $d\ge 3$.  Although it is generally believed that an
instance-optimal $\eps$-kernel or weak $\eps$-kernel in the plane can
be computed in polynomial time using dynamic programming, we are
unaware of any paper that presents such an algorithm. See below for
related work on this problem.  In this paper, we settle this question
by presenting \emph{fast} algorithms for computing instance-optimal
$\eps$-kernels and weak $\eps$-kernels for $d=2$.

\RegVer{\paragraph*{Related work.}}
\DCGVer{\subparagraph*{Related work.}}

As mentioned above, Agarwal \etal \cite{ahv-aemp-04} proved the
existence of an $\eps$-kernel of size $O(\eps^{-(d-1)/2})$ for any set
of points in $\reals^d$ and presented fast algorithms for computing
such an $\eps$-kernel. These algorithms were subsequently improved and
generalized, see \cite{c-fcscds-06,ay-sodsa-07,ahy-rsfpg-08}.  Yu
\etal~\cite{yapv-pmsfk-08} studied practical algorithms for computing
coresets/kernels, and suggested an incremental algorithm that seems to
provide a good approximation to the optimal kernel.

The NP-Hardness of computing an instance-optimal kernel in $\reals^3$
follows from that of polytope approximation~\cite{dg-cop3d-97}, see
also~\cite{akss-eakrm-17,clwww-krmse-17}.  Clarkson \cite{c-apca-93}
studied the problem of polytope approximation as a hitting-set
problem, providing a logarithmic approximation in the optimal size,
that can be used for approximating the optimal kernel.  For $d=3$, the
approximation factor can be improved to $O(1)$~\cite{bg-aoscf-95}.
Using a greedy approach, Blum \etal \cite{bhr-sagps-19} studied the
problem of approximating optimal kernels in high dimensions, and
presented polynomial-time algorithms for computing an $\eps$-kernel of
size $O(d\optX{\eps}\log \optX{\eps})$ or an
$(\eps+8\eps^{1/3})$-kernel of size $O(\optX{\eps}\eps^{-2/3})$.

More recently, there has been some work on computing variants of
$\eps$-kernels of minimum size, though \emph{none} of them compute an
instance-optimal $\eps$-kernel.  Wang \etal \cite{wmklt-mcmrm-21} use
a different definition of kernel, so comparing the results of this
paper to their work is somewhat confusing.  Specifically, Wang \etal
\cite{wmklt-mcmrm-21} presented a cubic-time algorithm that computes a
minimum-size subset $\QS$ of $\PS$ with the property that
\begin{math}
    \max_{p \in \PS} (1-\eps) \DotProd{v}{p} \le \max_{q \in \QS}
    \DotProd{v}{q},
\end{math}
assuming that $\PS$ is $\alpha$-fat for some constant $\alpha$; they
refer to such a subset as a $\eps$-core-set of $\PS$.  A shortcoming
of this definition is that it is neither translation nor
non-uniform-scaling invariant. However, it can be shown that their
algorithm computes an $\eps$-kernel of size at most $\optX{\eps/3}$
(observe that $\optX{\eps/3}$ can be much larger than $\optX{\eps}$).
Klimenko and Raichel \cite{kr-fechs-21} provided an $O(n^{2.53})$ time
algorithm for computing a minimum-size subset $\QS$ such that
$H(\CHX{\PS},\CHX{\QS})$, the Hausdorff distance between $\CHX{\PS}$
and $\CHX{\QS}$, is at most $\eps$.\footnote{%
   Recall that for two sets $A, B \in \reals^2$,
   $H(A,B) = \max\{h(A,B), h(B,A)\}$, where
   $h(X,Y) = \max_{x\in X}\min_{y\in Y} \dY{x}{y}$.}  They also tackle
the case when $\PS$ is convex, which they solve in $O(n \log^2 n)$
time.  The standard approach for computing small kernels, is to apply
an affine transformation to the point set to make it ``fat'', then
apply an algorithm for Hausdorff approximation, with parameter
$\eps/c$ where $c$ depends on the fatness of the mapped point set and
its diameter. Using the algorithm in~\cite{kr-fechs-21}, an
$\eps$-kernel of size at most $\optX{\eps/2}$ can be computed in
$O(n^{2.53})$ time. We note that since $\eps$ is an absolute error,
the size of Hausdorff-approximation can be $\Omega(n)$ in the worst
case.  If we set the error parameter to be $\eps \cdot \diamX{\PS}$,
then there exists an $\eps$-Hausdorff approximation $\QS$ of size
$O(\eps^{-(d-1)/2})$ but $\QS$ may not be an $\eps$-kernel since for a
direction $v\in\dirs$, $\lenX{\pwIY{v}{\QS}}$ maybe as small as
$\lenX{\pwIY{v}{\PS}}-\eps\diamX{\PS}$, while $\eps$-kernel requires
$\pwIY{v}{\QS} \supseteq (1-\eps)\pwIY{v}{\PS}$.  As such while the
width or minimum-enclosing-box of an $\eps$-kernel approximates that
of $\PS$, a Hausdorff approximation does not offer such a guarantee
and thus not always suitable for approximating extent measures of
$\PS$.

There is also some connection between our problem and minimum-link
distance and polygon approximation, see
\cite{gm-oacmn-90, ghms-apsml-93,  mp-mpe-08, ms-sapo-95, wc-fmvvd-86, w-fmnp-91} for some
relevant results.

\RegVer{%
\begin{table}[t]
    \centering
    \begin{tabular}{|c|c|l|l|}
      \hline%
      Kernel size
      &
        Running time
      &
        Ref
      &
        Remark
      \\
      \hline
      \hline
      $O(1/\sqrt{\eps}) \Bigr.$
      &
        $O(n + 1/\eps^{3/2})$
      &
        Agarwal \etal \cite{ahv-aemp-04}
      &
      \\
      \hline
      $O(1/\sqrt{\eps}) \Bigr.$
      &
        $O(n + 1/\eps^{1/2})$
      &
        Chan \cite{c-fcscds-06}
      &
      \\
      \hline
      $\leq \displaystyle\optX{\eps/3}$
      &
        $O(n^3)\Bigr.$
      &
        Wang \etal
        \cite{wmklt-mcmrm-21}
      &
      \\
      \hline
      $\leq \displaystyle\optX{\eps/2}$
      &
        $O(n^{2.53})\Bigr.$
      &
        Klimenko and Raichel
        \cite{kr-fechs-21}
      &
      \\
      \hline
      $ \optX{\eps}$
      &
        $O(\optX{\eps} n \log n)\Bigr.$
      &
        \thmref{opt-kernel}%
      &
      \\
      \hline
      $\leq \optX{\eps/4}$
      &
        $O(n\log n)\Bigr.$
      &
        \thmref{eps-kernel1}%
      &
      \\
      \hline
      \hline
      $\optWX{\eps}$
      &
        $O(n^2\log n)\Bigr.$
      &
        \thmref{weak:kernel}%
      &
        \emph{weak} $\eps$-kernel
      \\
      \hline
      \hline
      \multicolumn{4}{c}{Hausdorff distance at most $\eps$ $\Bigr.$}\\
      \hline
      $\optHX{\eps}$
      &%
        $O( n^{2.53} )\Bigr.$
      &
        Klimenko and Raichel
        \cite{kr-fechs-21}
      &
      \\
      \hline
      $2\optHX{\eps}$
      &%
        $O( n \log^2 n )\Bigr.$
      &
        Klimenko and Raichel
        \cite{kr-fechs-21}
      &
      \\
      \hline
      $\optHX{\eps}$
      &%
        $O( n \log^2 n )\Bigr.$
      &
        Klimenko and Raichel
        \cite{kr-fechs-21}
      &
        Input in convex position
      \\
      \hline
      $\optHX{\eps}$
      &%
        $O(\optHX{\eps} n \log n)\Bigr.$
      &
        \thmref{hausdorff:opt}%
      &
      \\
      \hline
      \hline
    \end{tabular}
    \medskip%
    \caption{Results on computing/approximating optimal kernel. Here
       $\optX{\eps} := \optX{\eps,\PS}$ is the size of the instance-optimal
       $\eps$-kernel of $\PS$. Similarly, $\optWX{\eps}$ is the size of the
       instance-optimal weak $\eps$-kernel of $\PS$ (which is potentially smaller).
       Finally, $\optHX{\eps}$ is the size of the instance-optimal
       Hausdorff approximation with distance $\eps$ -- it can be as
       large as $n$.  }
    \tbllab{results}
\end{table}%
}
\RegVer{\paragraph*{Our results.}}
\DCGVer{\subparagraph*{Our results.}}
Let $\PS$ be a set of $n$ points in $\reals^2$, and let $\eps>0$ be a
parameter. There are three main results in this paper:

\begin{compactenumA}
    \smallskip%
    \item \titleX{Optimal kernel.} %
    We present (in \secref{opt-kernel}) an
    $O(\optX{\eps} n\log n)$-time algorithm for computing an
    $\eps$-kernel of $\PS$ of size $\optX{\eps} := \optY{\eps}{\PS}$;
    recall that $\optX{\eps} = O(\eps^{-1/2})$.

    \smallskip%
    \item \titleX{Optimal weak kernel.}  We present (in
    \secref{weak:opt:kernel}) an $O(n^2\log n)$-time algorithm for
    computing a weak $\eps$-kernel of $\PS$ of size
    $\optWY{\eps}{\PS}$, the minimum size of a weak $\eps$-kernel of
    $\PS$. %
\end{compactenumA}

Our algorithm for computing the optimal kernel can be adapted to
computing an optimal Hausdorff approximation of $\CHX{\PS}$:
\smallskip%
\begin{compactenumA}[resume]
    \smallskip%
   \item \titleX{Optimal Hausdorff approximation.} %
   We present (in \secref{hausdorff}) an
   $O(\optHX{\eps} n\log n)$-time algorithm for computing a set
   $\QS \subseteq \PS$ of size $\optHX{\eps}$ such that
   $H(\CHX{\PS}, \CHX{\QS}) \le \eps$, where $\optHX{\eps}$ is the
   size of the minimum such subset.
\end{compactenumA}
\smallskip%
We obtain these results by reducing the computation of (weak) optimal
kernel to the following two covering problems, which are of
independent interest:
\smallskip%
\begin{compactenumI}%
    \item \titleX{Optimal arc cover.} Given a set $\Xi$ of $n$ arcs of
    the unit circle $\dirs$, compute its smallest subset that covers
    $\dirs$. Lee and Lee~\cite{ll-ccmp-84} had presented an
    $O(n\log n)$-time algorithm for this problem, which is optimal in
    the worst case.  Here we present a somewhat simpler algorithm with
    the same running time (see \secref{c:arc:cover}), which is more
    intuitive and which we adapt to the computation of weak kernels.

    \smallskip%

    \item \titleX{Optimal star cover.} %
    Given a polygon $\PP$ that is star shaped with respect to the
    origin $\origin$ and a set of lines $\LS$, compute a smallest
    subset of lines (i.e., cuts) in $\LS$ that separate $\origin$ from
    $\partial \PP$. Alternatively, this can be interpreted as covering
    $\partial\PP$ by the (outer) halfplanes defined by the lines of
    $\LS$.  We reduce this problem to the above arc-cover problem, but
    the number of candidate arcs can be quadratic.  We use a greedy
    algorithm to prune the number of candidate arcs to $O( \opt n)$,
    in $O( \opt n \log n)$ time, where $\opt$ is the size of the
    optimal solution, and then compute an arc cover in
    $O( \opt n \log n)$ time using the above algorithm.  We reduce the
    computation of $\eps$-kernel to this covering problem by using the
    polarity transform (see \secref{opt-kernel}).
\end{compactenumI}%
\medskip%
Finally, we introduce (in \secref{core}) the concept of core of a
point set, prove its properties, and describe an algorithm for
computing it.  A convex body $C$ can be represented as the
intersection of all the minimal slabs that contains it. The
\emphw{$\eps$-core} is the result of intersecting all these slabs
after shrinking them by a factor of $1-\eps$.  It induces an
affine-invariant inner approximation of $C$.  For a point set $\PS$,
its $\eps$-core is a convex polygon lying inside $\CHX{\PS}$.  We
describe an $O(n\log n)$-time algorithm for computing the $\eps$-core
of $\PS$.

We show that the convex hull of any $\eps$-kernel of $\PS$ contains
the $\eps$-core of $\PS$, and that any subset $\CS \subseteq \PS$
whose convex hull contains the $\eps$-core is a $4\eps$-kernel of
$\PS$, see \lemref{core:4:approx}. Thus the $\eps$-core is an
approximation to the optimal $\eps$-kernel, which has the benefit of
being well defined for any bounded convex shape.  We believe this
notion of $\eps$-core is new, and is of independent interest.  We
present an $O(n\log n)$-time algorithm for computing the smallest
subset of $\PS$ such that its convex-hull contains the $\eps/4$-core
of $\PS$, which yields an $\eps$-kernel of $\PS$ of size at most
$\opt_{\eps/4}$.

\RegVer{%
   \smallskip%
   \noindent
   Our results are summarized in \tblref{results}.
}

\DCGVer{\subparagraph{Summary}}%
\RegVer{\paragraph{Summary.}}%
We provide a near-linear-time algorithm for computing an
instance-optimal kernel in two dimensions, which as far as we know is
the first such algorithm. Previously only super-quadratic bicriteria
approximation algorithms were known for computing an instance-optimal
$\eps$-kernel (the only exception is the algorithm of Klimenko and
Raichel \cite{kr-fechs-21} that works for the special case that the
points are in convex position for Hausdorff distance). On the way, we
visit several interesting new problems, as detailed above, and present
efficient algorithms for them and reductions between them.

\section{Preliminaries}
\seclab{preliminaries}

Let $\PS$ be a set of $n$ points in $\reals^2$, and let
$\eps \in (0,1)$ be a parameter.  Without loss of generality assume
that the origin $\origin$ lies in the interior of $\CHX{\PS}$, where
$\CHX{\PS}$ denotes the convex-hull of $\PS$ (if
$\origin\not\in\CHX{\PS}$, one can choose three arbitrary points of
$\PS$ and translate $\PS$ so that their centroid becomes $\origin$).

A direction in $\reals^2$ can be represented as a unit vector in
$\reals^2$.  The set of unit vectors (directions) in $\reals^2$ is
denoted by $\dirs = \Set{ \pa \in \reals^2}{\!\normX{\pa}=1\bigr.}$.

\begin{definition} \deflab{h:n}%
    For a line $\Line$ not passing through the origin, let
    $\hn = \hnX{\Line}$ (resp. $\hp =\hpX{\Line}$) be the (closed)
    halfplane bounded by $\Line$ and containing (resp. not containing)
    the origin.

    For a direction $v\in \dirs$ and a point $\pb\in\reals^2$, let
    $\hnDirY{v}{\pb}$ be the halfplane that is bounded by the line
    normal to direction $v$ and passing through $\pb$, and that
    contains $\origin$.
\end{definition}

\begin{definition}[Extremal point, supporting line]
    \deflab{supporting:line}%
    For a direction $v\in\dirs$, let $\extX{v}$ be the
    \emphi{extremal point} of $\PS$ in the direction $v$. That is
    \begin{math}
        \extX{v} = \arg\max_{\pa\in\PS} \DotProd{v}{\pa}.
    \end{math}
    The point $\extX{v}$ is unique if $v$ is not the outer normal of
    an edge of $\CHX{\PS}$.  Similarly, let $\extLX{v}$ be the
    \emphw{supporting line} of $\CHX{\PS}$ normal to $v$ and passing
    through $\extX{v}$. Let $\hnDirX{v} = \hnX{\extLX{v}}$ and
    $\hpDirX{v} = \hpX{\extLX{v}}$.  Observe that
    $\CHX{\PS} \subset \hnDirX{v}$.

    For a real number $\psi$, let $\hnDirMY{v}{\psi}$ and
    $\hpDirMY{v}{\psi}$ be the halfplanes formed by translating
    $\hnDirX{v}$ and $\hpDirX{v}$, respectively, towards the origin by
    distance $\psi$.
\end{definition}

\begin{definition}
    \deflab{refinement}%
    The \emphi{normal diagram} of $\PS$ is the partition of $\dirs$
    into maximal intervals so that the extremal point $\extX{v}$
    remains the same for all directions within an interval. The
    endpoints of these intervals correspond to the outer normals of
    the edges of $\CHX{\PS}$.
    The normal diagram can be further refined so that for all
    directions $v$ within each interval, both $\extX{v}$ and
    $\pa_{-v}$ remain the same.  Such a pair of points are
    \emphw{antipodal pairs}. Let $\ndiag = \ndiag(\PS)$ denote this
    \emphi{refinement} of the normal diagram, and observe that
    $\cardin{\ndiag} \leq 2n$. See \figref{kernel}.
\end{definition}

\begin{figure}[ht]
    \begin{tabular}{*{3}{c}}
    \begin{minipage}{0.3\linewidth}
        \includegraphics[page=1,width=0.94\linewidth]%
        {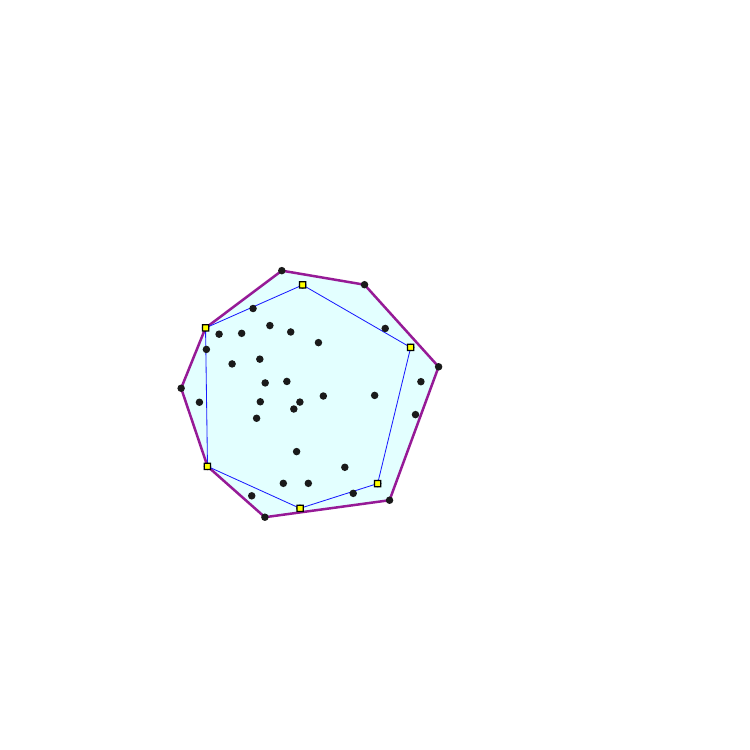}%
        \hfill%
    \end{minipage}
		&
    \begin{minipage}{0.3\linewidth}
        \includegraphics[page=2,width=0.94\linewidth]{figs/normal_diagram}%
    \end{minipage}
		&
    \begin{minipage}{0.3\linewidth}
        \includegraphics[page=3,width=0.44\linewidth]{figs/normal_diagram}
        \includegraphics[page=4,width=0.44\linewidth]{figs/normal_diagram}
    \end{minipage}\hfill
    \\
      (A)&(B)&(C)
    \end{tabular}
    \caption{(A) Point set $\PS$, $\CHX{\PS}$, and $\eps$-kernel of
       $\PS$ (say, for $\eps=0.2$). (B) Directions in which a point is
       extremal. (C) Normal diagram of $P$ and its refinement
       $\ndiag(\PS)$, $\pwIX{v}$.}
    \figlab{kernel}
\end{figure}

\DCGVer{\subparagraph*{Directional width and $\eps$-kernel.}}
\RegVer{\paragraph*{Directional width and $\eps$-kernel.}}
For a direction $v \in \dirs$, let
\begin{equation*}
    \pwIY{v}{\PS}%
    =%
    \Bigl[%
    \min_{p \in \PS} \DotProd{v}{p} ,
    \,
    \max_{p \in \PS} \DotProd{v}{p}
    \Bigr]
\end{equation*}
denote the \emphi{projection interval} of $\PS$ in direction $v$. Its
length $\pwY{v}{\PS} = \normX{ \pwIX{v} }$ is the \emphi{directional
   width} of $\PS$ in the direction of $v$. Note that
$\pwIX{v}=-\pwIX{-v}$ and $\pwY{v}{\PS}=\pwY{-v}{\PS}$.  For an
$\eps \in (0,1)$ and an interval $J = [x,y]$, let $(1-\eps)J$ be the
shrinking of $J$ by a factor of $(1-\eps)$, i.e.,
$(1-\eps)J = [x+(\eps/2)\lenX{J},y-(\eps/2)\lenX{J}]$.

\begin{definition}
    \deflab{approx:and:kernel}%
    A set $X \subseteq \CHX{\PS}$ is an \emphi{$\eps$-approximation}
    of $\PS$ if $\pwIY{u}{X} \supseteq (1-\eps)\pwIY{u}{\PS}$ for all
    directions $u\in\dirs$.  A subset $\CS \subseteq \PS$ is a
    ``strong'' \emphi{$\eps$-kernel} of $\PS$ if it is an
    $\eps$-approximation of $\PS$. Let $\opt_\eps(\PS)$ denote the
    minimum size of an $\eps$-kernel of $\PS$.  See
    \figref{kernel:example} for an example.
\end{definition}

\begin{figure}[h]
    \includegraphics[page=1,width=0.3\linewidth]{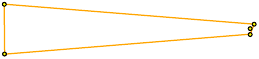}
    \hfill%
    \includegraphics[page=2,width=0.3\linewidth]{figs/opt_kernel}
    \hfill%
    \includegraphics[page=3,width=0.3\linewidth]{figs/opt_kernel}

    \phantom{}\hfill(A) \hfill (B) \hfill (C)\hfill \phantom{}

    \caption{ (A) A point set and its convex hull. (B) Its
       $0.2$-core. (C) Its optimal $0.2$-kernel -- observe that it
       contains points that are not on the convex-hull.}
    \figlab{kernel:example}
\end{figure}

We emphasize that the shrinking here is done for every direction
individually around the center of the projection interval -- in
particular, there is no center point of the $\CHX{\PS}$ around which
we do the scaling -- to some extent this gives rise to most of the
technical difficulties in constructing and approximating an optimal
kernel.  The following property of $\eps$-approximation will be useful
later on.
\begin{lemma}[\cite{ahv-aemp-04}]
    \lemlab{affine}%
    Let $\PS$ be a point set in $\reals^d$, $X\subseteq \CHX{\PS}$,
    and $T$ an affine map in $\reals^d$.  $X$ is an
    $\eps$-approximation for $\PS$ $\iff$ $T(X)$ is an
    $\eps$-approximation of $T(\PS)$.
\end{lemma}

A slightly weaker notion of $\eps$-kernel was used by Agarwal \etal
\cite{ahv-aemp-04}, that is potentially (significantly) smaller than
their ``strong'' counterparts but somewhat harder to compute.
\begin{definition}
    \deflab{weak:kernel}%
    A subset $\CS \subseteq \PS$ is a \emphi{weak $\eps$-kernel} of $\PS$
    if $\pwY{u}{\CS}\ge (1-\eps)\pwY{u}{\PS}$ for all
    $u\in\dirs$.
\end{definition}
This weaker definition was sufficient for the purposes of
Agarwal \etalNP. However, it is less intuitive than the
stronger variant, and it is harder to compute the optimal weak kernel.

\subsection{Computing optimal circular arc cover}
\seclab{c:arc:cover}

Let $\ICS$ denote a set of $n$ circular arcs on $\dirs$, each of
length less than $\pi$, that cover $\dirs$.\footnote{By computing the
   union of arcs in $\ICS$, we can decide, in $O(n\log n)$ time,
   whether $\ICS$ covers $\dirs$.}  As mentioned in the introduction,
an $O(n\log n)$-time algorithm for computing the smallest subset of
$\ICS$ that cover $\dirs$ was proposed in~\cite{ll-ccmp-84}.  Below we
present an alternative $O(n\log n)$ time algorithm for computing the
smallest-size arc cover from $\ICS$, which we believe is simpler and
more intuitive -- the indifferent reader is encouraged to skip the
rest of this section.

\medskip

The basic idea is to use the greedy algorithm. Picking a start arc,
and then going counterclockwise as far one can adding arcs in a greedy
fashion results in a cover of size $\opt+1$, where $\opt$ is the
optimal size.  After an $O(n\log n)$ preprocessing, the greedy
algorithm can be executed in $O(\opt)$ time.  To reduce the size of
the solution to $\opt$, one has to guess a starting arc that belongs
to the optimal solution. We show that the least covered point on the
circle is covered by $O(n/\opt)$ intervals.  This implies that one has
to try only $O( n/\opt)$ starting arcs and thus run the greedy
algorithm $O(n/\opt)$ times. The overall running time is thus
$O(n\log n)$.

\paragraph{Setup.} %
As a reminder, $\ICS$ is a set of $n$ circular arcs on $\dirs$, each
of length less than $\pi$, that cover $\dirs$.  For an arc
$I \in \ICS$, we refer to the clockwise endpoint of $I$ as its
\emphw{right} endpoint, denoted by $\clockwiseX{I}$, and its other
endpoint as its \emphw{left} endpoint.  For a pair of intersecting
arcs $I,J \in \ICS$, let $I \preceq J$ if the right endpoint of
$I\cap J$ is that of $I$ (the two arcs may have the same right
endpoint).

The endpoints of arcs in $\ICS$ partition $\dirs$ into a set $\Gamma$
of $2n$ atomic intervals. All points in an interval $\gamma\in\Gamma$
lie on the same subset of arcs of $\ICS$, denoted by
$\ICS_\gamma$. For each interval $\gamma\in \Gamma$, set
$n_\gamma = \cardin{\ICS_\gamma}$, and let
$\cwX{\gamma} := \arg \max_\preceq \ICS_\gamma$ denote the arc of
$\ICS_\gamma$ whose right endpoint is the most clockwise from
$\gamma$. For each arc $I\in\ICS$, we set $\rho(\gamma)$ to be the
atomic interval whose right endpoint is the right endpoint of $I$.

\paragraph*{The algorithm.}
The algorithm works as follows: %
\begin{compactenumI}
    \item \steplab{i:c:a} Sort the endpoints of $\ICS$ and compute
    $\Gamma$.
    \item Doing an angular sweep along $\dirs$ and maintaining the
    subset of $\ICS$ that contains the sweeping point, compute
    $n_\gamma$ and $I_\gamma$ for each interval
    $\gamma\in\Gamma$. (Since $\ICS$ covers $\dirs$, $n_\gamma\ge0$
    for all $\gamma\in\Gamma$.)
    \item \steplab{i:c:c} Compute
    $\gamma_0 = \arg\min_{\gamma\in\Gamma} n_\gamma$, the atomic
    interval that lies in the fewest input arcs.
    \item For each arc $I \in \ICS_{\gamma_0}$, compute a circular-arc
    cover $\CC_I$, using a greedy algorithm, starting from $I$, as
    follows: Set $I_0 := I$. At the $j$\th step, let
    $\gamma_j = \rho(I_{j-1})$ be the atomic interval whose right
    endpoint is the right endpoint of $I_{j-1}$.  Set
    $I_j := \cwX{\gamma_j}$ and $\CC_I := \CC_I \cup \{I_j\}$, and
    continue until arrive at an arc that contains the left endpoint of
    $I_0$.
    \item Among all the covers $\CC_I$, $I\in \ICS_{\gamma_0}$,
    computed in the previous step, return the smallest one.
\end{compactenumI}

\paragraph*{Correctness.}
The correctness of the algorithm follows from the following lemma,
which is well known, see e.g.~\cite{ii-oaapl-86}.

\begin{lemma}
    \lemlab{greedy}%
    Let $\ICS$ be a set of $n$ circular arcs on $\dirs$, each of
    length at most $\pi$, that cover $\dirs$. Let $\opt$ be the
    minimum size of a subset of $\ICS$ that covers $\dirs$. For each
    $I\in\ICS_{\gamma_0}$, $\cardin{\CC_I} \leq \opt+1$.  Furthermore,
    if $I$ belongs to an optimal solution, then
    $\cardin{\CC_I} = \opt$.
\end{lemma}
\begin{proof}
    It is well known that the greedy algorithm computes an optimal
    solution for covering an open interval from a given set of
    intervals~\cite{clrs-ia-09}.

    We note that $X=\dirs\setminus I_0$ is an open interval, and
    therefore the greedy algorithm computes an optimal cover of
    $X$. Obviously, $X$ has a cover of size at most $\opt$, so the
    greedy algorithm terminates after at most $\opt$ steps. Hence, it
    computes a cover of $\dirs$ of size at most $\opt+1$. Furthermore,
    if $I$ is in an optimal solution, then $X$ can be covered by
    $\opt-1$ arcs of $\ICS$ and thus the greedy algorithm computes a
    cover of $\dirs$ of size $\opt$ in this case.
\end{proof}

Since any cover of $\dirs$ has to contain an arc of $\ICS_{\gamma_0}$,
\lemref{greedy} implies that the above algorithm computes an optimal
cover of $\dirs$.  We now analyze the running time of the algorithm.
Steps \stepref{i:c:a}--\stepref{i:c:c} take $O(n\log n)$ time. By
\lemref{greedy}, each execution of the greedy algorithm takes
$O(\opt)$ time, so the total running time is
$O(n\log n+ n_{\gamma_0}\opt)$. It suffices to bound the value of
$n_{\gamma_0}$.

\begin{lemma}
    \lemlab{shallow:point}%
    Let $\ICS$ be a set of circular arcs of $\dirs$ that cover it, and
    let $\opt$ be the size of the minimum cover of $\dirs$ by the arcs
    of $\ICS$. Then there is a point $\pa \in \dirs$ that is contained
    only in $O( n/\opt)$ arcs of $\ICS$.
\end{lemma}
\begin{proof}
    The claim is immediate if $\opt \leq 10$. So assume $\opt >10$,
    and let $I_1, I_2, \ldots, I_{t}$ be the $t$ arcs computed by the
    greedy algorithm (here $t \leq \opt+1$). Let
    $\pa_i = \clockwiseX{I_{2i-1}}$, for
    $i=1, \ldots, \alpha = \floor{(t-2)/2}$. The greedy algorithm
    ensures that no arc of $\ICS$ that covers $\pa_i$ can cover
    $\pa_{i+1}$. Since $\alpha \geq (\opt-5)/2$, it follows that there
    is an index $j$, such that $\pa_j$ lies on at most
    $n/\alpha = O(n /\opt)$ arcs, which implies the claim.
\end{proof}

By \lemref{shallow:point}, $n_{\gamma_0} = O(n/\opt)$, and thus the
overall running time of the algorithm is $O(n\log n)$. Putting
everything together, implies the following:

\begin{theorem}
    \thmlab{arc-cover}%
    Let $\ICS$ be a set of $n$ circular arcs on $\dirs$.  The optimal
    cover of $\dirs$ by the arcs of $\ICS$, if there exists one, can
    be computed in $O(n \log n)$ time.
\end{theorem}

\section{Covering a star polygon by halfplanes}
\seclab{star-cover}

The input is a set of $\LS$ of $n$ lines and a polygon $\SP$ with
$O(n)$ vertices that is star-shaped with respect to the origin
$\origin$ (i.e., for every point $\pa \in \SP$,
$\origin \pa \subseteq \SP$).  Formally, the task at hand is to
compute a minimum set of lines $\CS \subseteq \LS$, such that for any
point $\pa \in \partial \SP$, $\interX{\origin \pa}$ intersects a line
of $\CS$. Geometrically,
$\faceoX{\CS} := \bigcap_{\Line\in\CS} \hnX{\Line}$, the intersection
of inner halfplanes bounded by lines in $\CS$, is contained in $\SP$.
An alternative interpretation of this problem is that $\partial\SP \subset \bigcup_{\Line\in\CS}\hpX{\Line}$.

\subsection{Reduction to  arc cover}
\seclab{arc-reduction}

$\partial \SP$ can be viewed as the image of a function
$\SP: \dirs \rightarrow \reals^2$.  Specifically, for a direction
$u \in \dirs$, $\SP(u)$ is the intersection point of $\partial\SP$
with the ray from the origin in direction $u$.
A line $\Line$ \emphi{blocks} the direction $u$ if $\Line$ intersects
the segment $\origin\SP(u)$. A subset $\GS\subseteq\LS$ is a
\emphi{blocking set} of $\SP$ if each direction in $\dirs$ is blocked
by at least one line of $\GS$ (i.e., $\faceoX{\GS}\subset\SP$).

Fix a line $\Line\in\LS$. Let $\CTY{ \Line}{ \SP}$ denote the set of
connected components (i.e., segments) of $\Line\cap \SP$. For a
segment $\seg \in \CTY{ \Line}{ \SP}$, let
\begin{math}
    \intervalX{\seg}=\Set{
       \smash{{\origin\pa}/{\normX{\origin\pa}}}\in\dirs }{ \pa\in \seg}
\end{math}
be the circular arc induced by $\seg$.  All directions in
$\intervalX{\seg}$ are blocked by $\Line$.  Let
$\intervalX{\Line}=\Set{\intervalX{\seg}}{\seg \in \CTY{\Line}{\SP}}$
be the set of all circular arcs that are induced by blocking segments
of $\Line$.  Let $\Xi= \bigcup_{\Line\in\LS} \intervalX{\Line}$ be the
set of all circular arcs defined by the lines of $\LS$.  For a subset
$\Gamma\subseteq\Xi$, let
$\LS(\Gamma) = \{ \Line \in \LS \mid \gamma\in\intervalX{\Line},
\gamma\in\Gamma\}$ be the original subset of lines of $\LS$ supporting
the arcs of $\Gamma$.

\begin{lemma}
    \lemlab{kernel-arc}%
    (i) If $\Gamma \subseteq\Xi$ is an arc cover, i.e.,
    $\bigcup \Gamma = \dirs$, then $\LS(\Gamma)$ is a blocking set.

    (ii) There is an arc cover $\Gamma \subseteq \Xi$ of size $k$ if
    and only if there is a blocking set $\GS\subseteq\LS$ of size $k$.
\end{lemma}
\begin{proof}
    (i) If $\Gamma$ is an arc cover, then for every direction
    $u\in\dirs$, there is an arc $\intervalX{\seg} \in \Gamma$ that
    blocks the direction $u$. If
    $\intervalX{\seg} \in \intervalX{\Line}$, for a line
    $\Line \in \LS(\Gamma)$, then the segment $\origin\SP(u)$
    intersects $\Line$.  Since this condition holds for all directions
    in $\dirs$, it follows $\LS(\Gamma) \subseteq \LS$ is a blocking
    set.

    (ii) If there is an arc cover $\Gamma\subseteq\Xi$ of size $k$,
    then by part (i), $\LS(\Gamma)$ is a blocking set of size at most
    $k$. Conversely, let $\GS$ be a blocking set for $\SP$.  Without
    loss of generality, we can assume that each line of $\GS$ appears
    as an edge on the boundary of the face $\face$ of $\ArrX{\GS}$
    that contains the origin, because otherwise we can remove the line
    from $\GS$. For each line $\Line \in \GS$, let
    $\seg_\Line \in \CTY{ \Line}{ \SP}$ be the segment that contains
    the edge of $\face$ lying on $\Line$.  Since $\face \subseteq\SP$,
    the segment $\origin\SP(u)$ intersects an edge of $\face$ for
    every $u\in\dirs$. Hence,
    $\Set{\intervalX{\seg_\Line} }{ \Line\in \GS}$ is an arc cover of
    size at most $\cardin{\GS}$.
\end{proof}

By \lemref{kernel-arc}, it suffices to compute the smallest-size arc
cover from $\Xi$. But $\cardin{\Xi} = \Theta(n^2)$ in the worst
case. Therefore computing $\Xi$ explicitly and then using
\thmref{arc-cover} to compute an arc cover take $O(n^2\log n)$
time. In the following, we show how to improve the running time to
$O(n\opt \log n)$, where $\opt$ is the optimal solution size.

\subsection{Computing an almost-optimal blocking set}
\seclab{almost-optimal}

We extend the greedy algorithm used in the circular arc cover (see
\secref{c:arc:cover}) to compute an arc cover in $\Xi$ without
computing $\Xi$ explicitly. For clarity, we describe the greedy
algorithm in terms of computing a blocking set.

For a pair of directions $u, v\in \dirs$, let $\SP(u,v]\subseteq\SP$
be the semiopen subchain of $\SP$ from $\SP(u)$ to $\SP(v)$ in the
counterclockwise direction, which contains the endpoint $\SP(v)$ but
not $\SP(u)$. As such, we have $\SP(u,u] =\SP$.

We define a (partial) function
$\segfn: \dirs\times\LS \rightarrow \Re^4$, as follows.  For a pair
$u\in\dirs$ and a line $\Line\in\LS$, if $\Line$ does not intersect
the segment $\origin\SP(u)$, then $\segfnY{u}{\Line}$ is not
defined. Otherwise, it is the segment of $\CTY{\Line}{\SP}$ that
intersects $\origin \SP(u)$. Similarly, we define a (partial) function
$\shootfn: \dirs\times\LS \rightarrow\dirs$, that is the first point
of $\segfnY{u}{\Line}$ in the counter-clockwise direction after
$\SP(u)$ (note, that $\Line$ might intersect the boundary $\SP$ many
times).  Set $\lambda(u) = \arg\max_{\Line\in\LS} \shoot{u}{\Line}$,
i.e., among the feasible segments that intersect $\origin\SP(u)$,
$\lambda(u)$ is the last one to exit $\SP$ in the counterclockwise
direction.

The algorithm consists of the following steps: Set $v_0:=(1,0)$,
$\Line_0:=\lambda(v_0)$, $\GS:=\{\Line_0\}$, and $i:=1$. In the $i$\th
iteration, the algorithm does the following: it sets
$v_i =\shoot{v_{i-1}}{\Line_{i-1}}$, $\Line_i = \lambda (v_{i})$, and
$\GS = \GS \cup \{\Line_i\}$. The algorithm then continues to the next
iteration till $\faceoX{\GS} \subseteq \interX{\SP}$.  Let $v'_1$ be
the first intersection point of $\Line_0$ with $\SP$ in the clockwise
direction from $v_0$, i.e., the segment $\SP(v_1)\SP(v'_1)$ lies
inside $\SP$. Then the terminating condition is the same as
$\shoot{v_i}{\Line_i}$ lying after $v'_1$ (from $v_i$) in the
counterclockwise direction.  By construction,
$\faceoX{\GS}\subset \SP$. Since this is a greedy algorithm for
computing an arc cover, $\cardin{\GS} \leq \opt+1$. The polygon $\SP$
can be preprocessed, in $O(n\log n)$ time, into a data structure of
linear size so that for a pair $u\in\dirs$ and a line $\Line\in\LS$,
$\shoot{u}{\Line}$ can be computed in $O(\log n)$
time~\cite{cg-vip-89,hs-parss-95}.  The algorithm performs $O(n\opt)$
such queries, so the total running time is $O(n\opt\log n)$.

\begin{lemma} %
    \lemlab{opt-hs-1}%
    Let $\LS$ be a set of $n$ lines, $\SP$ be a polygon with $O(n)$
    vertices that is star shaped with respect to $\origin$ and that
    contains $\faceoX{\LS}$, and let $\opt$ be the size of the
    smallest blocking set in $\LS$ for $\SP$.  A blocking set
    $\GS \subseteq \LS$ of size at most $\opt+1$ can be computed in
    $O( \opt n \log n )$ time.
\end{lemma}

\subsection{Computing an optimal solution}
\seclab{optimal:i:c}

Let $\GS$ be the blocking set computed by the above greedy algorithm.
For each line $\Line \in \LS$ we compute its intersection points with
the lines of $\GS$.  For each such intersection point $\xi$, if $\xi$
lies inside $\SP$, let $\seg_\xi \in \CTY{\PP}{\Line}$ be the segment
that contains $\xi$. Let $\SS_1$ be the set of resulting $O(n\opt)$
segments. Let
\begin{equation*}
    \SS_2 =  \bigcup_{\Line \in \GS} \CTY{\Line}{\PP}
\end{equation*}
be the set of all segments induced by the lines of $\GS$.  Set
$\SS = \SS_1 \cup \SS_2$. The computes the set
$\Gamma = \Set{\intervalX{\seg}}{\seg \in \SS }$, and then computes
the minimal size arc cover $\CS$ of $\dirs$ by the arcs of $\Gamma$.
The returned set is
$\OS=\{\Line\in\Line \mid \intervalX{\seg}\in\CS, \seg\subset\Line\}$.

\begin{figure}[ht]
    \centering%
    \begin{tabular}{*{3}{c}}
      \includegraphics[page=1,width=0.4\linewidth]{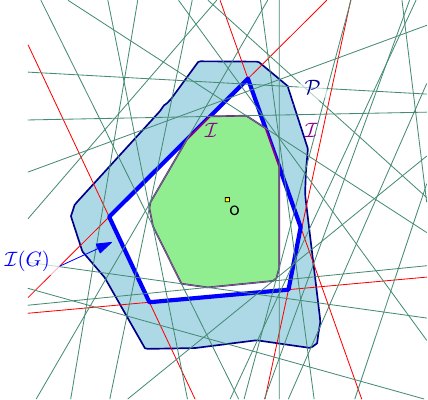}
      &
        \qquad
      & \includegraphics[page=2,width=0.4\linewidth]{figs/filtering}
      \\
      (A)&&(B)
    \end{tabular}
    \caption{(A) A hitting set $\GS$ of size at most $\opt+1$.  (B)
       Illustration of the proof of
       \protect\lemref{reduced-arc-cover}.}
    \figlab{filtering}
\end{figure}

\begin{lemma}
    \lemlab{reduced-arc-cover}%
    The set $\Gamma$ contains an arc cover of size $\opt$.
\end{lemma}
\begin{proof}
    Suppose for the sake of contradiction that $\Gamma$ does not
    contain an arc cover of size $\opt$.  Let $\CS \subset \Xi$ be an
    arc cover of size $\opt$. Then $\CS$ contains an arc
    $\intervalX{\seg}$ such that $\seg$ lies on a line of
    $\LS\setminus\GS$ and $\seg$ does not intersect any line of $\GS$,
    i.e., it lies in the interior of a face of $\ArrX{\GS}$, the
    arrangement of $\GS$.

    If $\seg$ lies in the face corresponding to $\faceoX{\GS}$, then
    $\seg$ must intersect $\partial \faceoX{\GS}$, as the endpoints of
    $\seg$ lies on $\partial \SP$ and $\IS(\GS) \subseteq \SP$,
    contradicting the assumption that $\seg$ does not intersect any
    line of $\GS$.

    Next, suppose $\seg$ lies in some other face of $\ArrX{\GS}$. Let
    $\pa$ be an endpoint of $\seg$.  The segment $\pa \origin$ must
    intersect a line $\Line' \in \GS$ at a point $\pb$. In particular,
    let $\seg' \in \Line' \cap \SP$ be the segment of $\Line$
    containing $\pb$. Clearly, $\seg'$ is a blocker for all the points
    on $\seg$, so we can obtain another optimal solution by replacing
    $\intervalX{\seg}$ with $\intervalX{\seg'}$ (see
    \figref{filtering}), and this solution has one more arc of
    $\Gamma$, a contradiction.

    Hence, we can conclude that $\Gamma$ contains an optimal arc
    cover.
\end{proof}

Computing the set $\GS$ takes $O(n \opt \log n)$ time.  Observe that
$\cardin{\SS_1} = O(n \opt)$, as each line of $\LS$ induces at most
$\opt+1$ segments in this set.  Similarly, as $\cardin{\GS} = \opt+1$,
we have that $\cardin{\SS_2} = O(n \opt)$. It follows that computing
$\SS_1$ and $\SS_2$ requires $O(n \opt)$ ray-shooting queries in
$\SP$, and these queries overall take $O(n \opt \log n)$ time. Hence,
we obtain the following:

\begin{lemma}
    \lemlab{opt-hs}%
    Let $\LS$ be a set of $n$ lines in the plane, and let $\SP$ be a
    polygon with $O(n)$ vertices that is star shaped with respect to
    $\origin$ and that contains $\faceoX{\LS}$.  Then a blocking set
    from $\LS$ of $\SP$ of size $\opt$ can be computed, in
    $O(\opt n \log n)$ time, where $\opt$ is the size of the optimal
    solution.
\end{lemma}

\section{Computing optimal \TPDF{$\eps$}{eps}-kernel}
\seclab{opt-kernel}

Let $\PS$ be a set of $n$ points in $\reals^2$ and
$\eps\in (0,1)$ a parameter.  We describe an
$O(n\opt_\eps\log n)$-time algorithm for computing an $\eps$-kernel of
size $\opt_\eps$. We use polarity to construct a set $\LS$ of $n$ lines and a star polygon $\SP$ that
contains $\faceoX{\LS}=\bigcap_{\Line\in\CS} \hnX{\Line}$. An $\eps$-kernel of $\PS$ corresponds to a blocking set in $\LS$ for $\SP$.

\begin{definition}[$\eps$-shifted supporting line]
    \deflab{eps:supporting:line}%
    For a direction $u\in\dirs$ and a parameter $\eps>0$, let
    $\LineSY{u}{\eps}$ be the boundary line of
    $\hnSY{u}{\eps} = \hnDirMY{u}{ (\eps/2) \pwY{u}{\PS} }$, see
    \defref{supporting:line}. Let $\hpSY{u}{\eps}$ be the (closed)
    complement halfplane to $\hnSY{u}{\eps}$.
\end{definition}

Set $\HH_\eps=\Set{\smash{\hpSY{u}{\eps}}}{ u \in \dirs}$.  The
following lemma is immediate from the definition of $\eps$-kernel.

\begin{lemma}
    Given a point set $\PS$ in $\reals^2$ and a parameter
    $\eps\in(0,1)$, a subset $\CS\subseteq\PS$ is an $\eps$-kernel of
    $\PS$ if and only if $\hpSY{v}{\eps} \cap \CS \ne \emptyset$ for
    all $u\in\dirs$, i.e., $\CS$ is a hitting set of $\HH_\eps$.
\end{lemma}

The problem of computing an $\eps$-kernel thus reduces to computing a
minimum-size hitting set of the infinite set $\HH_\eps$. It will be
convenient to use the polarity transform and work in the
mapped plane, so we first describe the polar of $\eps$-kernel and
then describe the algorithm.

\RegVer{\paragraph*{Polarity.}}
\DCGVer{\subparagraph*{Polarity.}}

For a point $\pa \neq \origin$, its \emphi{inversion}, through the
unit circle, is the point $\pa^{-1} = \pa / \normX{\pa}^2$. Observe
that $\pa, \pa^{-1}, \origin$ are collinear,
$\normX{\pa} \normX{\pa^{-1}} =1$, and $\pa$ and $\pa^{-1}$ are on the
same side of the origin on this line.  We use the \emphw{polarity
   transform}, which maps a point $\pa=(a,b) \neq \origin$ to the
line
\begin{equation*}
    \polarX{\pa}%
    \equiv%
    ax+by-1=0
    \equiv%
    \DotProd{\pa}{ (x,y) } - 1 = 0
    \equiv%
    \DotProd{\pa}{ (x,y) - \frac{\pa}{\normX{\pa}^2} }  = 0.
\end{equation*}
Namely, the line $\polarX{\pa}$ is orthogonal to the vector
$\origin \pa$, and the closest point on $\polarX{\pa}$ to the origin
is $\pa^{-1}$. Geometrically, a point $\pa$ is being mapped to the line
passing through the inverted point $\pa^{-1}$ and orthogonal to the
vector $\origin \pa^{-1}$. Similarly, for a line $\Line$, its
\emphw{polar} point $\polarX{\Line}$ is $\pb^{-1}$, where $\pb$ is the
closest point to the origin on $\Line$. Observe that
$\polarX{\pth{\polarX{\Line}}} = \Line$ and
$\polarX{\pth{\polarX{\pa}}} = \pa$ for any line $\Line$ and any point
$\pa$.

\begin{figure}[ht]
    \phantom{}%
    \hfill%
    \begin{minipage}{0.35\linewidth}
        \includegraphics{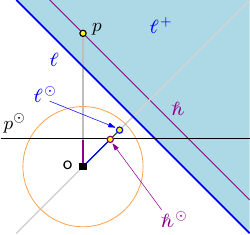}\hfill
    \end{minipage}
    \hfill%
    \begin{minipage}{0.35\linewidth}
        \includegraphics[width=0.99\linewidth]{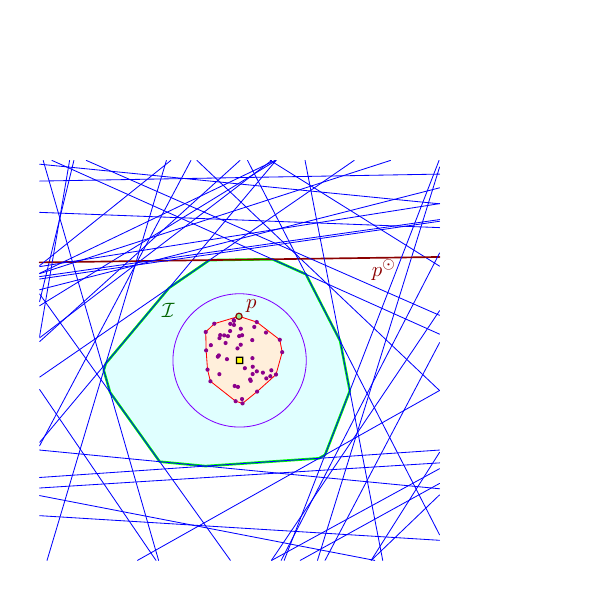}
    \end{minipage}
    \hfill%
    \phantom{}%

    \captionof{figure}{%
       Left: A point $\pa$ lies in the halfplane $\hpX{\Line}$ $\iff$
       $\polarX{\pa}$ intersects the segment $\origin
       \polarX{\Line}$. %
       \\
       Right: A convex hull of a point set, and the corresponding
       ``polar'' polygon formed by the intersection of halfplanes.}

    \figlab{polarity:1}%
\end{figure}

If a point $\pa$ lies on a line $\Line$ then
$\polarX{\Line}\in \polarX{\pa}$. If $\pa$ lies in the halfplane
$\hpX{\Line}$ (by \defref{h:n}, we have that $\hpX{\Line}$ does not
contain $\origin$) if and only if $\polarX{\pa}$ intersects the
segment $\origin \polarX{\Line}$, see \figref{polarity:1} (left).  Set
$\polarX{\PS}=\Set{\polarX{\pa} }{\pa\in\PS}$ and
$\faceo := \faceoX{\polarX{\PS}} = \bigcap_{\pa\in \PS}
\hnX{\polarX{\pa}}$.  Then the polygon $\faceo$ is the polar of
$\CHX{\PS}$, namely:
\begin{compactenumI}[leftmargin=1.3cm]
    \medskip%
    \item If $\pa \in \PS$ is a vertex of $\CHX{\PS}$ then
    $\polarX{\pa}$ contains an edge of $\faceo$, see
    \figref{polarity:1} (right).

    \smallskip%
    \item The polar of line $\Line$ missing (resp. intersecting)
    $\CHX{\PS}$ is a point lying in (resp.\ out) $\faceo$.

    \smallskip%
    \item For a point $\pa \in \CHX{\PS}$,
    $\faceo \subset \hnX{\polarX{\pa}}$.
\end{compactenumI}%
\medskip%

Consider any direction $u\in\dirs$. Let $\pa_u$ be the extremal point
of $\PS$ in direction $u$, and let $\Line_u$ be the corresponding
supporting line, see \defref{supporting:line}.  The point
$\polarX{\Line_u}$ lies on the edge of $\faceo$ supported by
$\polarX{\pa_u}$, and $\polarX{\Line_u} / \normX{\polarX{\Line_u}}=u$.
Similarly, the polar of the shifted supporting line $\LineSY{u}{\eps}$
(see \defref{eps:supporting:line}), is the point
$\polarX{\LineSY{u}{\eps}}$ which lies outside $\faceo$ on the ray
induced by $u$ (starting at the origin).

\begin{figure}[ht]
    (A) \quad%
    \includegraphics[page=1,width=0.35\linewidth]{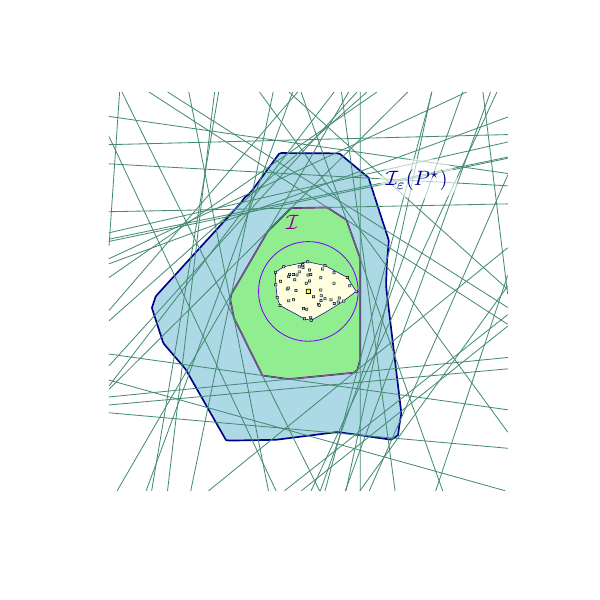}%
    \hfill%
    (B)%
    \quad%
    \includegraphics[page=2,width=0.35\linewidth]{figs/pos}\\
    \hfill%
    \phantom{}%
    \caption{(A) $\CHX{\PS}$, $\IS(\polarX{\PS})$, %
       $\IS_\eps(\polarX{\PS})$.  (B) $\eps$-kernel $C$ and its polar
       $\polarX{\CS}$;
       $\CHX{\PS}\subseteq\IS{\polarX{\PS}}
       \subseteq\IS_\eps(\polarX{\CS})$. %
    }
        \figlab{polar-ker}
\end{figure}

\RegVer{\paragraph*{Kernel and polarity.}}
\DCGVer{\subparagraph*{Kernel and polarity.}}

Returning to $\eps$-kernels, let $\ndiag$ be the refinement of the
normal diagram of $\CHX{\PS}$, see \defref{refinement}.  Recall that $\ndiag$ is centrally symmetric.
The supporting lines $\Line_u$ and $\Line_{-u}$ support the same pair of
vertices of $\CHX{\PS}$ for all directions $u$ lying inside an
interval of $\ndiag$.
For each interval $\gamma \in \ndiag$, let $-\gamma$ denote
its antipodal interval. For each interval $\gamma\in\ndiag$, let
$\pa_\gamma$ be the supporting vertex of $\CHX{\PS}$ for all
directions in $\gamma$.

Let
$\sPntY{\pa}{\gamma,\eps} =
(1-\eps/2)\pa_\gamma+(\eps/2)\pa_{-\gamma}$.  It can be verified that
the line $\Line_{v,\eps}$ for $v\in\gamma$
passes through $\sPntY{\pa}{\gamma,\eps}$. Therefore the polar of the
set of lines $\Set{\Line_{v,\eps}}{ v \in \gamma }$ is a segment
$e_\gamma$ that lies on the line $\polarX{(\sPntY{\pa}{\gamma,\eps})}$
and outside $\faceo$.  The sequence
$\langle e_\gamma \mid \gamma\in\Gamma\rangle$ forms the boundary of a
polygon $\IS_\eps (\polarX{\PS})$ that is star shaped with respect to
$\origin$ and that contains $\faceo$ in its interior.  See
\figref{polar-ker}.  Putting everything together, we obtain the
following lemma, which characterizes the $\eps$-kernel after polarity.

\begin{lemma}
    \lemlab{polar-kernel}%
    Let $\PS$ be a set of $n$ points in $\reals^2$ and $\eps\in(0,1)$
    a parameter. The star-shaped polygon $\IS_\eps(\polarX{\PS})$ can
    be computed in $O(n \log n)$ time.  Furthermore, a subset
    $\CS\subseteq\PS$ is an $\eps$-kernel of $\PS$ if and only if
    $\polarX{\CS}$ is a blocking set for $\IS_\eps(\polarX{\PS})$
    (see \figref{polar-ker}).
\end{lemma}

Computing the smallest set $\CS \subseteq \PS$ thus reduces to the
star-polygon-cover problem. Using \lemref{opt-hs} and that there is an
$\eps$-kernel of size $O( \eps^{-1/2})$~\cite{ahv-aemp-04}, we obtain
the following:

\begin{theorem}
    \thmlab{opt-kernel}%
    Let $\PS$ be a set of $n$ points in $\reals^2$, and let
    $\eps \in (0,1)$ a parameter.  An optimal $\eps$-kernel of $\PS$
    of size $\opt$ can be computed in $O( \opt n \log n)$ time.  In
    the worst case, $\opt = O(\eps^{-1/2})$, and the running time is
    $O( \eps^{-1/2} n \log n)$.
\end{theorem}

\subsection{Quadratic lower bound}
Below we show that there exists a set $\PS$ of points such that there
are quadratic number of intersections between $\polarX{\PS}$ and
$\IS_\eps(\polarX{\PS})$.  This suggest that our somewhat more
involved algorithm (that is using the greedy algorithm to prune the
set of arcs used) is necessary even in this case. It will be more
convenient to use the duality transform instead of polarity for
describing the lower-bound construction.

\RegVer{\paragraph*{Duality and $\eps$-kernel.}} %
\DCGVer{\subparagraph*{Duality and $\eps$-kernel.}} %
The duality transform provides a similar mapping to polarity.  The
\emphic{dual point}{dual!point} to the line $\Line \equiv y= ax + b$
is the point $\dualX{\Line} = (a, -b)$. Similarly, for a point
$\pa = (c,d)$ its \emphic{dual line}{dual!line} is
$\dualX{\pa} \equiv y = c x - d$. Namely, for $\pa = (a,b )$, the dual
line is
\begin{math}
    \dualX{\pa} \equiv y = a x - b,
\end{math}
and for a line
\begin{math}
    \Line \equiv
    y = c' x + d' %
\end{math}
the dual point is
\begin{math}
    \dualX{\Line} = (c', -d').
\end{math}
The following interpretation of kernels in the dual is standard, and
goes back to the original work of Agarwal \etal \cite{ahv-aemp-04}. As
such, we state the problem in these settings without proving the
equivalence.

For a set of lines
$\LS = \dualX{\PS} = \Set{\dualX{\pa}}{\pa \in \PS}$ in the plane
(i.e., $\LS$ is a set of affine functions from $\Re$ to $\Re$), let
\begin{equation*}
    \UY{\LS}{x} = \max_{f \in \LS} f(x)
    \qquad\text{and}\qquad%
    \LY{\LS}{x} = \min_{f \in \LS} f(x),
\end{equation*}
be the upper and lower envelopes of $\LS$, respectively.  The function
$\UX{x}$ is convex, while $\LX{x}$ is concave.  The \emphi{extent} of
$\LS$ is
\begin{equation*}
    \eY{\LS}{x} = \UY{\LS}{x} - \LY{\LS}{x}.
\end{equation*}
For a fixed $\eps \in (0,1)$, the \emphw{$\eps$-upper envelope} and
\emphw{$\eps$-lower envelope} are
\begin{align*}
  &\uY{\LS}{x}%
    =%
    \UY{\LS}{x} - \frac{\eps}{2} \eY{\LS}{x}%
    =%
    \Bigl(1-\frac{\eps}{2}\Bigr)\UY{\LS}{x} + \frac{\eps}{2}  \LY{\LS}{x}\\
  \text{and} \qquad
  &  \lY{\LS}{x}%
    =%
    \LY{\LS}{x} + \frac{\eps}{2} \eY{\LS}{x}%
    =%
    \frac{\eps}{2}\UY{\LS}{x} + \Bigl(1-\frac{\eps}{2}\Bigr)  \LY{\LS}{x},
\end{align*}
respectively. Unfortunately, these functions are not necessarily
convex, as demonstrated in \figref{counter:example}.

Computing an optimal $\eps$-kernel for $\PS$ is equivalent to
computing a set of lines $\LSA \subseteq \LS$, such that
$ \UY{\LSA}{x}$ lies above $\uY{\LS}{x}$ (and of course below
$\UY{\LS}{x}$), for all $x$. And similarly, $\LY{\LSA}{x}$ lies below
$\lY{\LS}{x}$, for all $x$.

\begin{figure}
    \centerline{%
       \includegraphics[page=4]%
       {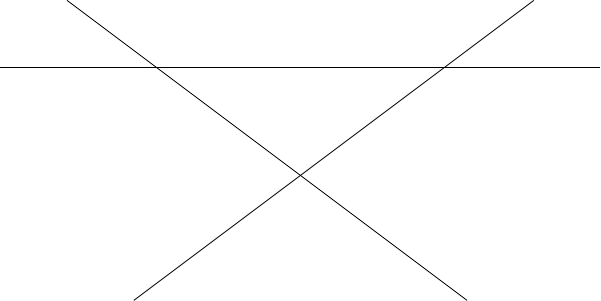}%
    }%
    \caption{Lower and upper envelopes, and their $\eps$-approximations.}
    \figlab{counter:example}
\end{figure}

\RegVer{\paragraph*{Lower-bound construction.}} %
\DCGVer{\subparagraph*{Lower-bound construction.}} %
Here we show that in the worst case the set
$\bigcup_{\Line \in \LS} (\CTY{ \Line}{ \PP})$ can have quadratic
size.
In particular, we construct a set of lines $\LS$, where the lines of
$\LS$ have quadratic number of intersections with $\uX{\cdot}$ and
$\lX{\cdot}$.

Consider the parabolas $f(x) = \frac{2}{\eps}(x^2 + 1)$ and
$g(x) = -\frac{1}{1-\eps / 2}(x^2 + 1)$.  Fix parameters $n$ and
$\eps$.  Let $\pa_i = \bigl(i/2n, f(i/2n)\bigr)$ and
$\pb_i = \bigl(i/2n,g(i/2n)\bigr)$, for $i=0, \ldots, 2n$. For a pair of distinct points $p, q\in\reals^2$, let $\Line(p,q)$ denote the line passing through $p$ and $q$. Let
\begin{align*}
  \DCGVer{&}
            \LS_f =
            \Set{ \Line( \pa_i, \pa_{i+2}) }{ i=0,2, 2n-2 }
            \DCGVer{\\}
  \qquad\text{and}\qquad%
  \DCGVer{&}
            \LS_g%
            =%
            \Set{ \Line( \pb_i, \pb_{i+2}) }{ i=1,3, 2n-3 }.
\end{align*}
The upper envelope of $\LS_f$ in the range $[0,1]$ is above $f(x)$,
except for touching it at the points $\pa_0, \pa_2, \ldots,
\pa_{2n}$. Similarly, the lower envelope of $\LS_g$, in the range
$I = [1/2n,1-1/2n]$ lies below $g$, except for touching it at the
points $\pb_1, \pb_3, \ldots, \pb_{2n-1}$.

It is easy to verify that the lines of $\LS_f$ and $\LS_g$ do not
intersect each other in the range $x \in [0,1]$. As such, the upper
envelope (resp. lower envelope) of $\LS = \LS_f \cup \LS_g$ in this
range is realized by the upper envelope (resp. lower envelope) of
$\LS_f$ (resp. $\LS_g$).

Consider a value $x \in \{ 1/2n, 3/2n, \ldots, (2n-1)/2n \}$. We have
that $\UY{\LS}{x} > f(x)$ and $\LY{\LS}{x}= g(x)$. As such, we have
\begin{align*}
  \uY{\LS}{x}%
  &=%
    \frac{\eps}{2}\UY{\LS}{x} + \Bigl(1-\frac{\eps}{2}\Bigr)  \LY{\LS}{x}
    \bm{>}%
    \frac{\eps}{2}f(x) + \Bigl(1-\frac{\eps}{2}\Bigr)  g(x)%
    \DCGVer{\\&}
  =%
  x^2 + 1 - (x^2 + 1)%
  =%
  0.
\end{align*}
Similarly, for $x \in \{ 2/2n, 4/2n, \ldots, (2n-2)/2n \}$, we have
$\UY{\LS}{x} = f(x)$ and $\LY{\LS}{x}< g(x)$. As such, we have
\begin{align*}
  \uY{\LS}{x}%
  &=%
    \frac{\eps}{2}\UY{\LS}{x} + \Bigl(1-\frac{\eps}{2}\Bigr)  \LY{\LS}{x}
    \bm{<}%
    \frac{\eps}{2}f(x) + \Bigl(1-\frac{\eps}{2}\Bigr)  g(x)%
    =%
    0.
\end{align*}

We thus obtain the following.

\begin{lemma}
    For any $\eps > 0$ and for any $n\ge 1$, there exists a set of
    $2n$ lines in $\reals^2$ whose $\eps$-upper envelope
    crosses the $x$-axis at least $2n-2$ times.
\end{lemma}

Next, we replicate the $x$-axis by sufficiently close (almost
parallel) $n$ lines that lie between the lower and upper envelopes of
$\LS$, and we add them to $\LS$. Then there are $\Omega(n^2)$
intersection points between $\uZ{\LS}$ and the lines of $\LS$.  We
thus get the following result.

\begin{lemma}
    \lemlab{lower-bound}%
    There exists a set $\LS$ of $n$ lines in $\reals^2$ such that the number of
    intersection points between $\IS_\eps(\LS)$ and $\LS$ is
    $\Omega(n^2)$.%
\end{lemma}

\section{Optimal Weak Kernel}
\seclab{weak:opt:kernel}

The above results dealt with the stronger notion of a kernel, but the
original work of Agarwal~\etal~\cite{ahv-aemp-04} defined a weaker
notion of a kernel, see \defref{weak:kernel}. In this section, we
present an $O(n^2\log n)$-time algorithm for computing an optimal weak
$\eps$-kernel, by reducing it to computing a smallest arc cover, with
some additional properties, in a set of $O(n^2)$ \emph{unit arcs}
(i.e., arcs on the unit circle).

Let $\PS$ be a set of $n$ points in $\reals^2$ and $\eps \in (0,1)$ a
parameter.  We parametrize $\dirs$ with the orientation in the range
$[-\pi,\pi]$ (with the two endpoints of this interval being glued together),
and let $u(\theta)=(\cos\theta,\sin\theta)$.
Recall
that a subset $\CS\subseteq \PS$ is an weak $\eps$-kernel of $\PS$ if
\begin{equation}
    \eqlab{weak-ker}
    \pwY{u(\theta)}{\CS}\ge (1-\eps)\pwY{u(\theta)}{\PS}
\end{equation}
for all $\theta\in[-\pi,\pi]$. Since
$\pwY{u(\theta)}{\PS}=\pwY{u(-\theta)}{\PS}$, it suffices to satisfy
\Eqref{weak-ker} for the angular interval $[-\pi/2,\pi/2]$. However,
it will be convenient to work with the entire $\dirs$, so let
\begin{equation*}
    \eY{\PS}{\theta}
    =
    \pwY{u(\theta/2)}{\PS} \quad \text {for } \theta\in[-\pi,\pi].
\end{equation*}
A subset $\CS\subseteq \PS$ is a \emphi{weak $\eps$-kernel} if and only if
\[
    \forall \theta\in [-\pi,\pi]%
    \quad%
    \eY{\CS}{\theta} \geq (1-\eps)\eY{\PS}{\theta} .
\]
For a pair $1\leq i<j\leq n$ and $\theta\in[-\pi,\pi]$, we define
$\gamma_{ij} \in [-\pi,\pi] \rightarrow \reals_{\ge 0}$ as
\begin{equation*}
    \gamma_{ij}(\theta)
    :=
    \cardin{\DotProd{u(\theta/2)}{\pa_i-\pa_j}}
    =%
    \cardin{(a_i-a_j)\cos (\theta/2)+(b_i-b_j)\sin (\theta/2)},
\end{equation*}
where $p_i=(a_i,b_i)$. Set
$\Gamma=\{\gamma_{ij}\mid 1 \leq i < j \leq n\}$. It is easily seen that
$\UY{\Gamma}{\theta}=\eY{\PS}{\theta}$.  For a pair $1\leq i<j\leq n$,
we define %
\begin{equation*}
    I_{ij}
    =%
    \Set{ \theta \in [-\pi,\pi]}{\gamma_{ij}(\theta) \geq
       (1-\eps)\UZ{\Gamma}(\theta)}.
\end{equation*}

\begin{lemma}
    \lemlab{weak-interval}
    The set $I_{ij}$ is a single connected circular arc.
\end{lemma}

\begin{proof}
    It is convenient to reparameterize $\gamma_{ij}$. More precisely,
    we define the function
    $\tgamma_{ij}: \reals \rightarrow \reals_{\ge 0}$ as
    \begin{equation}
	\tgamma_{ij}(x) =
        \abs{(a_i-a_j) + (b_i-b_j)x}
        \quad \mbox{for $x\in\reals$.}
    \end{equation}
    Set
    \begin{equation*}
        \TG = \Set{ \tgamma_{ij} }{ 1 \leq i < j \leq n}.
    \end{equation*}
    The analog of $I_{ij}$ is the set
    \[
        \tI_{ij}
        =%
        \Set{x\in\reals}{\tgamma_{ij}(x) \geq
           (1-\eps)\UY{\TG}{x}}.\Bigr.
    \]
    Note that $\cos (\theta/2) \geq 0$ for $\theta \in [-\pi,\pi]$,
    and as such
    \begin{align*}
      &
        \frac{1}{\sqrt{1+\tan^2
        \frac{\theta}{2}}}\tgamma_{ij}\bigl(\tan \tfrac{\theta}{2}\bigr)
        =%
        \cos \tfrac{\theta}{2}
        \tgamma_{ij}\bigl(\tan \tfrac{\theta}{2}\bigr)
        =%
        \vabs{(a_i-a_j) + (b_i-b_j)
        \tan \tfrac{\theta}{2}}
        \cos {\tfrac{\theta}{2}}
      \\&\qquad%
      =%
      \vabs{(a_i-a_j)\cos {\tfrac{\theta}{2}} + (b_i-b_j)
      \sin \tfrac{\theta}{2}}
      = %
      \gamma_{ij}(\theta).
    \end{align*}
    Observe that
    \begin{align*}
      \tau = \tan (\theta/2)\in \tI_{ij}%
      &\iff%
        \tgamma_{ij}(\tau) \geq
        (1-\eps)\UY{\TG}{\tau}
      \\&%
      \iff%
      \tgamma_{ij}(\tau) \geq
      (1-\eps)\max_{\tgamma \in \TG} \tgamma(\tau)%
      \\&%
      \iff%
      \frac{1}{\sqrt{1+\tau^2 }}
      \tgamma_{ij}(\tau) \geq
      (1-\eps)\max_{\tgamma \in \TG} \frac{1}{\sqrt{1+\tau^2
      }} \tgamma(\tau)%
      \\&%
      \iff%
      \gamma_{ij}(\theta)
      \geq
      (1-\eps)\max_{\gamma \in \Gamma}  \gamma(\theta)%
      \\&%
      \iff%
      \gamma_{ij}(\theta)
      \geq
      (1-\eps)\UZ{\Gamma}(\theta)
      \\&
      \iff%
      \theta\in I_{ij}.
    \end{align*}

    The graph of $\tgamma_{ij}$ is a cone with axis of symmetry around
    the $y$-axis and apex on the $x$-axis -- specifically, there are
    two numbers $\alpha_{ij}, \beta_{ij}$ such that
    $\tgamma_{ij}(x) = \alpha_{ij}\cardin{ x- \beta_{ij}}$. The number
    $\alpha_{ij}$ is the \emphi{slope} of $\tgamma_{ij}$. The function
    $(1-\eps)\UZ{\TG}$ is a convex chain, which is the upper envelope
    of the functions $(1-\eps)\tgamma_{ij}$, see
    \figref{weak-interval} (A).

    \begin{figure}[h]
        \includegraphics[page=1,width=0.45\linewidth]{\si{figs/weak_ker}}
        \hfill
        \includegraphics[page=2,width=0.45\linewidth]{\si{figs/weak_ker}}%
        \\[-0.4cm]
        \phantom{}\hfill(A)\hfill\hfill(B)\hfill\phantom{}
       \caption{Illustration of the proof of
          \lemref{weak-interval}. (A) Upper envelope $\UXX{\TG}$,
          and lower-bound curve $(1-\eps)\UXX{\TG}$. (B) A cone
          with higher slope ``buries'' at least one leg of the other
          cone.}
       \figlab{weak-interval}
   \end{figure}

   The graph of $\tgamma_{ij}$ is composed of two rays.  The set
   $\tI_{ij}$ is the (projection of the) intersection of the graph of
   $\tgamma_{ij}$ with a convex region. Thus $\tI_{ij}$ is potentially
   the union of two intervals (potentially infinite rays).  If
   $\tI_{ij}$ does not contain any finite interval, i.e., consists of
   two rays, then $I_{ij}$ is a single arc containing the orientation
   $\pi$. So assume that $\tI_{ij}$ contains a finite interval, see
   \figref{weak-interval} (B).  This implies that there are indices
   $u,v$, such that $(1-\eps) \tgamma_{uv}$ has higher slope than
   $\tgamma_{ij}$. But then $(1-\eps)\tgamma_{uv}$ is completely above
   one of the two rays forming the image of $\tgamma_{ij}$, implying
   that $\tI_{ij}$ can only be a single interval in this case. This in
   turn implies that $I_{ij}$ consists of a single arc.
\end{proof}

\RegVer{\paragraph{Computing the set $\IS$.}}
\DCGVer{\subparagraph*{Computing the set $\IS$.}} %
As a reminder $\IS = \Set{I_{ij}}{1 \leq i < j \leq n}$. To compute
it, we work in the line space, first computing the function
$\UY{\TG}{\cdot}$. This can be done in $O(n \log n)$ time by computing
the upper and lower envelopes of the lines of $\dualX{\PS}$, and then
merging them to get the extent function, which is
$\UY{\TG}{\cdot}$. We then shrink it down to get the function
$f(x) = (1-\eps)\UY{\TG}{x}$. Given a function
$\tgamma_{i,j} \in \TG$, to compute $\tI_{ij}$, we need to compute the
intersection of its two rays with the graph $f(x)$, which is a convex
polygonal line of complexity $O(n)$. This can be done using standard
techniques in $O( \log n)$ time per query. Mapping $\tI_{ij}$ to the
angular space results in $I_{i,j}$. Doing this for all $i,j$, results
in the set $\IS$. The overall running time is $O(n^2 \log n)$.

\RegVer{\paragraph{A $2$-approximation algorithm.}}
\DCGVer{\subparagraph*{A $2$-approximation algorithm.}} %
 Using the algorithm of
\thmref{arc-cover}, we compute, in $O(n^2\log n)$ time, a minimum arc
cover $\JS \subseteq \IS$.  Each interval $I_{ij}\in \JS$ corresponds
to two points $p_i, p_j$ of $\PS$. Set
$\CS := \{p_i,p_j \mid I_{ij}\in\JS \}$.
\begin{lemma}
    \lemlab{weak-cover}%
    $\CS$ is an weak $\eps$-kernel of size at most twice the optimal
    size.
\end{lemma}
\begin{proof}
    Since $\JS$ is an arc cover, for any $\theta\in[-\pi,\pi]$, there
    is pair $\pa_i,\pa_j\in\CS$ such that
    $\gamma_{ij}(\theta) \ge (1-\eps)\UY{\Gamma}{\theta}$. Therefore
    $\eY{\CS}{\theta} \ge (1-\eps)\UY{\Gamma}{\theta}
    =(1-\eps)\eY{\PS}{\theta}$, implying that $C$ is a weak
    $\eps$-kernel.

    Conversely, let $\CS^*$ be an optimal weak $\eps$-kernel. We
    construct an arc cover $\JS^*$ as follows.  The points in $\CS^*$
    are in convex position.  Consider $\ndiag = \ndiag(\CS^*)$ the
    refined normal diagram of $\CS$, which is a centrally symmetric
    partition of $\dirs$ into $2\cardin{\CS^*}$ intervals such that
    each pair of antipodal intervals of is associated with an
    antipodal pair of points $p_i, p_j\in\CS^*$. For each such pair
    $p_i,p_j$, we add the interval $I_{ij}$ to $\JS^*$;
    $\cardin{\JS^*} = \cardin{\CS^*}$.  For $\theta\in [-\pi,\pi]$,
    suppose $p_i,p_j$ is the supporting pair in directions
    $u(\theta/2)$ and $-u(\theta/2)$, respectively.  Then
    $\gamma_{ij}(\theta) = \eY{\CS^*}{\theta}$.  Since
    $\eY{\CS^*}{\theta} \ge (1-\eps)\UY{\Gamma}{\theta}$,
    $\theta\in I_{ij} \in \JS^*$.  Hence, $\JS^*$ is an arc cover.

    We can thus conclude that $\cardin{\CS} \leq 2\cardin{\CS^*}$.
\end{proof}

\RegVer{\paragraph{An exact algorithm.}}
\DCGVer{\subparagraph*{An exact algorithm.}}

The above algorithm is a $2$-approximation because it uses two
potentially new points for each interval. We can change the arc-cover
problem to account for this. We label every arc in $\IS$ by two
indices $i,j \in \IRX{n}$ -- indices of the pair of points in $\PS$
that define it.  An arc cover $\JS\subset\IS$ of $\dirs$ is
\emphi{admissible} if every pair of intersecting arcs in $\JS$ share
exactly one label. For any admissible arc cover $\JS$, the size of the
set $\Set{p_i,p_j}{I_{ij}\in\JS}$ is at most
$\cardin{\JS}$. Furthermore, the arc cover constructed from a weak
kernel in the proof of \lemref{weak-cover} is admissible. Therefore it
suffices to compute a minimum-size admissible arc cover in $\IS$.

To compute the smallest admissible arc cover, we follow the ideas in
the algorithm of for the arc-cover \secref{c:arc:cover}. While
$\cardin{\IS} = O(n^2)$, there must be a direction $u \in \dirs$ that
is covered by at most $O( n^2/\opt)$ intervals of $\IS$, where $\opt$
is the size of the optimal weak $\eps$-kernel. Let $\JS \subseteq \IS$
be the set of intervals covering $u$ ($u$ and $\JS$ can be computed in
$O(n^2 \log n)$ time). For each one of these intervals, we now perform
the greedy algorithm, as in \secref{c:arc:cover}. The only difference
is that instead of having a global data structure for all intervals,
we break them into $n$ groups. Specifically, for $i=1,\ldots, n$, let
$\IS_i\subset \IS$ be the set of all arcs $I$ with $i$ being one of
the two indices in its label.  Now, we build the necessary
data-structure used in \secref{c:arc:cover} for each such group. Now,
if the current interval is $I_{ij}$, the algorithm uses the
data-structures for $\IS_i$ and $\IS_j$ to generate two candidate
intervals to be used by the greedy algorithm. The algorithm uses the
one that extends further clockwise. The rest of the algorithm is the
same as in \secref{c:arc:cover}. This algorithm computes the smallest
admissible circular arc cover $\JS^*$. We return the set
$\Set{p_i,p_j}{I_{ij}\in\JS^*}$, which in view of the above discussion
is an optimal weak $\eps$-kernel.  Putting everything together we
obtain the following:

\begin{theorem}
    \thmlab{weak:kernel}%
    Given a set $\PS$ of $n$ points in the plane and a parameter
    $\eps \in (0,1)$, an optimal weak $\eps$-kernel
    of $\PS$ can be computed in $O(n^2\log n)$ time.
\end{theorem}

\section{Core of a point set}
\seclab{core}

We now define the core of a point set and show how it can be used to construct an $\eps$-kernel of small size.

\subsection{Definition and properties}
\seclab{core:def}

For a direction $v \in\dirs$, we define $\slabY{v}{\PS}$ to be the
minimum slab bounded by two lines that are orthogonal to $v$, and that
contains $\PS$, i.e., $\slabY{v}{\PS}= \hnDirX{v} \cap
\hnDirX{-v}$. Equivalently,
\begin{equation*}
    \slabY{v}{\PS}%
    =%
    \Set{\pa \in \Re^2}{\! \DotProd{\pa}{v} \in \pwIY{v}{\PS} }.
\end{equation*}
The shrunken copy of $\slabY{v}{\PS}$ is
\begin{equation*}
    (1-\eps)\slabY{v}{\PS}%
    =%
    \Set{\pa \in \Re^2}{\! \DotProd{\pa}{v} \in (1-\eps)\pwIY{v}{\PS} }.
\end{equation*}

For a parameter $\eps \in (0,1)$, the \emphi{$\eps$-core} of $\PS$ is
the set
\begin{equation*}
    \coreY{\PS}{\eps}%
    =%
    \bigcap_{v \in \dirs} (1-\eps)\slabY{v}{\PS}.
\end{equation*}
See \figref{core} and \figref{core:2} for examples.

\begin{figure}
    \phantom{} \hfill%
    \includegraphics[page=1,width=0.3\linewidth]{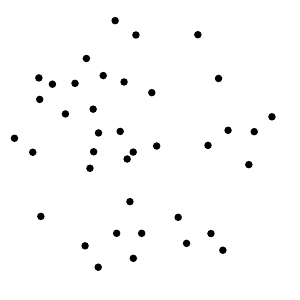}
    \hfill%
    \includegraphics[page=2,width=0.3\linewidth]{figs/eps_core}
    \hfill%
    \includegraphics[page=3,width=0.3\linewidth]{figs/eps_core}
    \hfill%
    \phantom{}
    \caption{A point set, a shrunken slab, and the resulting
       $0.2$-core.}
    \figlab{core}
\end{figure}

\begin{figure}
    \phantom{} \hfill%
    \includegraphics[page=1]{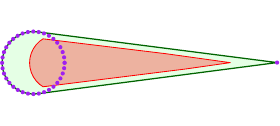}%
    \hfill%
    \includegraphics[page=2]{figs/core_2}%
    \hfill%
    \phantom{}
    \caption{Left: A point set and its $\eps$-core, for
       $\eps=0.2$. Right: The $\eps$-core is strictly contained inside
       a set that is an $\eps$-kernel in this case, as the projection
       interval on the $x$-axis that the $0.1$-kernel needs to contain
       is wider than the $\eps$-core projection interval.}
    \figlab{core:2}
\end{figure}

For an affine map $T$ in $\reals^2$, it is known that
$T(\coreY{\PS}{\eps}) = \coreY{T(\PS)}{\eps}$, see
\lemref{affine}. The following property of $\coreY{\PS}{\eps}$ follows
from its definition:
\begin{lemma}
    \lemlab{core:4:approx}%
    For any $\eps>0$ and for any $\eps$-kernel $\QS\subseteq\PS$,
    $\coreY{\PS}{\eps} \subseteq \CHX{\QS}$.
\end{lemma}

As demonstrated by \figref{core:2}, the exact converse of the above
lemma does not hold.  The next lemma proves a weaker converse property
of $\coreY{\PS}{\eps}$.
\begin{lemma}
    \lemlab{core:approx}%
    Let $\PSA$ be a set of $n$ points in $\reals^2$.  The polygon
    $\coreY{\PSA}{\eps}$ is a $4\eps$-approximation of $\PSA$ (see
    \defref{approx:and:kernel}).
\end{lemma}

\begin{proof}
    Let $\ellipse$ be the largest area ellipse inside $\CHX{\PSA}$. By
    John's ellipsoid theorem~\cite{h-gaa-11},
    $\CHX{\PSA} \subseteq 2\ellipse$.  Consider the affine map $T$
    that maps $\ellipse$ to $\diskY{\origin}{1}$, the disk of radius
    $1$ centered at the origin $\origin$. Let $\PS = T(\PSA)$. Then
    $\diskY{\origin}{1} \subseteq \CHX{\PS} \subseteq
    \diskY{\origin}{2}$. By \lemref{core:4:approx}, we have
    $U = \coreY{\PS}{\eps} = T(\coreY{\PSA}{\eps})$.

    For any direction $v\in\dirs$, consider the halfplane
    $\hnDirX{v}$ and its translation $\hnDirMY{v}{2\eps}$ (see
    \defref{supporting:line}). Observe that
    $\CHX{\PS}\subseteq \hnDirX{v}$ and
    $\diskY{\origin}{1} \subseteq \hnDirX{v}$.  We have that
    \begin{equation*}
        \pth{\hnDirMY{v}{2\eps}} \cap
        \pth{\hnDirMY{-v}{2\eps}}
        \subseteq%
        (1-\eps)\slabY{v}{\PS},
    \end{equation*}
    as the slab $(1-\eps)\slabY{v}{\PS}$ is formed by translating
    $\hnDirX{v}$ and $\hnDirX{-v}$ by distance at most
    $(\eps/2)\diamX{\PS} \leq (\eps/2) 4 \leq 2\eps$. Hence,
    \begin{equation*}
        L%
        =%
        \bigcap_{v \in \dirs} \pth{\hnDirMY{v}{2\eps}}
        \subseteq%
        \bigcap_{v \in \dirs} (1-\eps)\slabY{v}{\PS}
        =%
        U.
    \end{equation*}

    We claim that the distance of any point $\pa \in \CHX{\PS}$ from
    $U$ is at most $4\eps$.  If
    $\diskY{\pa}{2\eps} \subseteq \CHX{\PS}$, then
    $\pa \in \hnDirMY{v}{2\eps}$ for all $v$, which implies that
    $\pa \in L \subseteq U$. This argument implies that
    $\diskY{\origin}{1-2\eps} \subseteq U$. As such, any point
    $\pa \in \CHX{\PS} \cap \diskY{\origin}{1+2\eps}$ is in distance
    at most $4\eps$ from $U$.

    \begin{figure}[h]
        \phantom{}\hfill%
        \includegraphics[page=1,width=0.45\linewidth]{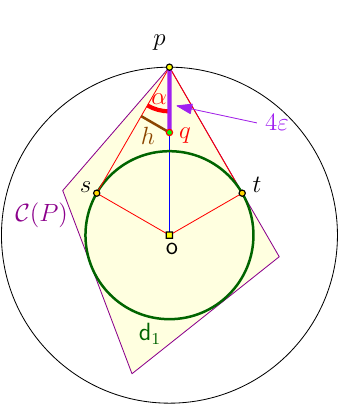}%
        \hfill%
        {\includegraphics[page=2,width=0.45\linewidth]{figs/good_approx}}
        \hfill\phantom{}%
        \caption{Illustration of the proof of
           \protect\lemref{core:approx}.}
        \figlab{proof}
    \end{figure}%

    Let $\pa$ be any point in
    $\CHX{\PS} \setminus \diskY{\origin}{1+2\eps}$.  Let $\pb$ be the
    point at distance $4\eps$ from $\pa$ along the segment
    $\pa \origin$, see \figref{proof}. Let $\pc$ and $\pd$ be the two
    points of $\diskY{\origin}{1}$ at which the tangent from $p$
    touches $\diskY{\origin}{1}$.  As $\dY{\pa}{\origin} \leq 2$, the
    angle $\alpha = \angle \origin\pa\pc \ge 30^\circ$. As such, the
    distance of $\pb$ from $\pa\pc$ is at least
    $4\eps \sin \alpha \geq 2\eps$. We conclude that
    $\diskY{\pb}{2\eps} \subseteq \CHX{\PS}$. This implies that
    $\pb \in L \subseteq U$, which implies the claim, as
    $\dY{\pa}{\pb} =4\eps$.

    Consider any direction $v$, and consider the two points $\pa, \pb$
    that are extreme in $\PS$ in direction $\pm v$. Let $\pa'$ and
    $\pb'$ be the nearest-points to $\pa$ and $\pb$, respectively, in
    $U$. Since projection (on a direction) is a contraction, we have
    that
    \begin{math}
        \cardin{ \DotProd{\pa - \pa'}{v} }%
        \leq%
        \cardin{ \pa - \pa' }%
        \leq%
        4\eps.
    \end{math}
    Similarly, $\abs{\DotProd{\pb - \pb'}{v}} \leq 4\eps$. We conclude
    that
    \begin{equation*}
        \pwY{v}{U}%
        \geq
        \abs{\DotProd{\pa' - \pb'}{v}}%
        \geq%
        \abs{\DotProd{\pa - \pb}{v}} - 8\eps%
        =%
        \pwY{v}{\PS} -8 \eps
        \geq (1-4\eps)
        \pwY{v}{\PS},
    \end{equation*}
    as $\pwY{v}{\PS} \geq \pwY{v}{\diskY{\origin}{1}} \geq 2$.
\end{proof}

\begin{remark}
    The bound of $1-4\eps$ on the approximation factor is most likely
    not tight. Proving a sharp bound on the approximation factor is an
    interesting problem for further research.
\end{remark}

\begin{corollary}
    \corlab{core-ker}%
    Let $\eps \in (0,1)$ be a parameter, and let $\CS \subseteq \PS$
    be a set of points such that $\coreY{\PS}{\eps}\subset\CHX{\CS}$.
    Then $\CS$ is a $4\eps$-kernel of $\PS$.
\end{corollary}

Many of the known algorithms for computing an $\eps$-kernel of a point
set $\PS$ basically compute a set $C$ of points such that
$\coreY{\PS}{\eps/4}\subset\CHX{\CS}$~\cite{ahv-aemp-04,c-fcscds-06}.

\subsection{Computing the core}
\seclab{compute:core}

We now describe an algorithm for computing $\coreY{\PS}{\eps}$. Our
algorithm will also imply that $\coreY{\PS}{\eps}$ is a convex polygon
with at most $2n$ vertices.

We first compute $\CC = \CHX{\PS}$ the convex-hull of $\PS$. After
having computed $\CC$, we compute the refinement $\ndiag(\PS)$ of the
normal diagram of $\CC$.  Recall that for all direction $u$ in an
interval of $\ndiag(\PS)$, the boundary lines of $\slabY{u}{\PS}$
passes through the same antipodal pair.

For an interval $I=[u,v]$ of $\ndiag(\PS)$, observe that
$R_I := \bigcap_{w \in I} (1-\eps)\slabY{w}{\PS}$ is the parallelogram
$(1-\eps)\slabY{u}{\PS}\cap (1-\eps)\slabY{v}{\PS}$. Indeed, let
$\pa'$ (resp., $\pb'$) be the point along $\pa \pb$ at distance
$(\eps/2)\dY{\pa}{\pb}$ from $\pa$ (reps., $\pb$), i.e.,
$\pa' = (1-\tfrac\eps2)\pa+\tfrac\eps2 \pb$ and
$\pb' = \tfrac\eps2\pa+(1-\tfrac\eps2)\pb$. Now, for any $w \in I$,
the two lines bounding $(1-\eps)\slabY{w}{\PS}$ pass through the
points $\pa',\pb'$. As such,
$R_I =(1-\eps)\slabY{u}{\PS}\cap (1-\eps)\slabY{v}{\PS}$. See
\figref{shrank}. Hence, $\coreY{\PS}{\eps} = \bigcap_{I\in\ndiag(\PS)} R_I$.

\begin{figure}
    \centerline{\includegraphics{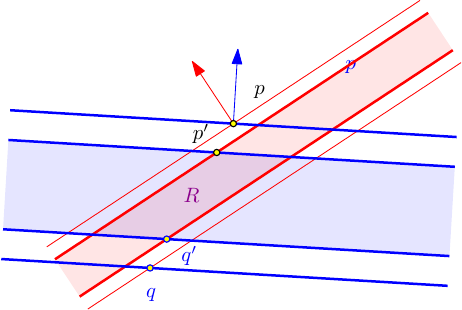}}
    \caption{The rotating calipers and the intersection of the
       shrunken slabs.}
    \figlab{shrank}
\end{figure}%

These shrunken slabs $(1-\eps)\slabY{v}{\PS}$ can be computed for each
endpoint $v$ of an interval of $\ndiag(\PS)$, each of which is an
intersection of two halfplanes. Next, we compute the intersection of
these at most $2n$ halfplanes in $O(n \log n)$
time~\cite{bcko-cgaa-08}. (We note that though $\ndiag$ has $2n$
endpoints, $\slabY{v}{\PS}=\slabY{-v}{\PS}$, so there are only $2n$
halfplanes.)  However, since these halfplanes are computed in sorted
order of their normals (by the rotating caliper algorithm), the
intersection can be computed in linear time. The rotating caliper
algorithm itself runs in linear time once $\CHX{\PS}$ is provided. We
thus obtain the following:

\begin{lemma}
	\lemlab{core:algo}
    Given a set $\PS$ of $n$ points in the plane and a parameter
    $\eps \in (0,1)$, $\coreY{\PS}{\eps}$ is a convex polygon with at
    most $2n$ vertices, and it can be computed in $O(n \log n)$
    time. The running time can be improved to $O(n)$ if $\CHX{\PS}$ is
    provided.
\end{lemma}

Finally, we describe an algorithm for computing the smallest subset of
$\PS$ that contains $\coreY{\PS}{\eps}$.

\subsection{Minimal approximation of a convex polygon}
\seclab{reduction}

Let $\PS$ be a set of $n$ points in $\reals^2$, and let
$\InPoly \subseteq \CHX{\PS}$ be a convex polygon with $m$ vertices.
The task at hand is to compute a minimum set of points
$\CS \subseteq \PS$ such that $\InPoly \subseteq \CHX{\CS}$.  Without
loss of generality, assume that $\origin\in\interX{\InPoly}$ and that
no point of $\PS$ lies in the interior of $\InPoly$, as such a point
cannot be part of $\CS$.

\begin{figure}[h]
    \phantom{}%
    \hfill%
    (A)%
    \quad%
    \includegraphics[scale=0.8]{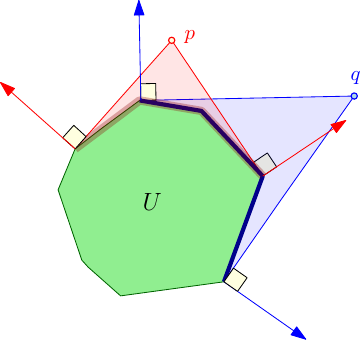} \hfill%
    (B)%
    \quad%
    \includegraphics[page=2]{figs/visible_interval} \\
    \hfill%
    \phantom{}
    \caption{%
       (A) Angular intervals $\AIX{\pa}$ and $\AIX{\pb}$.
       (B) The arc set $\ICS$.
    }
    \figlab{interval:graph}
\end{figure}

For a point $\pa \in \PS$, its \emphw{angular interval} is defined to
be
$\AIX{\pa} = \Set{ v \in \dirs}{ \InPoly \subset \hnDirY{v}{\pa} }$,
the set of directions for which the halfplane $\hnDirY{v}{\pa}$
contains $\InPoly$ (see \defref{h:n}).
The endpoints of $\AIX{\pa}$ are outer normals to
the two tangents of $\InPoly$ from $\pa$. Set
$\ICS = \Set{ \AIX{\pa}}{ \pa \in \PS}$. The set $\ICS$ can be
computed, in $O(n\log m)$ time, by computing the tangents to $\InPoly$
from each point of $\PS$, see \figref{interval:graph}.

\begin{lemma}
    \lemlab{cover:reduction}%
    For a subset $\CS$ of $\PS$, $\CHX{\CS} \supseteq \InPoly$ if and
    only if $\,\ICS_\CS=\Set{\AIX{\pa}}{ \pa \in \CS}$ covers
    $\,\dirs$.
\end{lemma}

\begin{proof}
    Let $C \subseteq \PS$ be such that $\CHX{\CS}\supseteq
    \InPoly$. Consider a direction $v\in\dirs$. Let $\gnDirX{v}$ be
    the supporting halfplane of $\InPoly$ with $v$ as its outer
    normal. Since $\InPoly \subseteq \CHX{\CS}$, there is a point
    $\pa \in \CS$ that does not lie in the interior of
    $\gnDirX{v}$. Then
    $\InPoly\subset \gnDirX{v} \subseteq \hnDirY{v}{\pa}$, implying
    that $v \in \AIX{\pa}$.  That is, the direction $v$ is
    covered. This implies that $\ICS_\CS$ covers $\dirs$.  See
    \figref{cover:by:intervals} (A).

    Conversely, suppose $\ICS_\CS$ covers $\dirs$. If
    $\CHX{\CS}\not\supset \InPoly$, then there is a vertex $\pb$ of
    $\InPoly$ that lies outside $\CHX{\CS}$, and thus $\pb$ is a
    vertex of $\CHX{\CS\cup \InPoly}$. Let $v$ be a direction such
    that $\hnDirY{v}{\pb}$ is a supporting halfplane of
    $\CHX{\CS\cup \InPoly}$. Since $\pb \not\in \CHX{\CS}$, the set
    $\CHX{\CS}$ lies in the interior of
    $\hnDirY{v}{\pb}$. Consequently, there is no point $\pa \in \CS$
    for which $\InPoly \subset \hnDirY{v}{\pa}$, implying that
    $v \not\in \bigcup \ICS_\CS$. This, however, contradicts the
    assumption that $\ICS_\CS$ covers $\dirs$. Hence,
    $\InPoly \subseteq \CHX{\CS}$.  See \figref{cover:by:intervals}
    (B).
\end{proof}

    \begin{figure}[h]
        \phantom{}\hfill(A)\hfill%
        \includegraphics{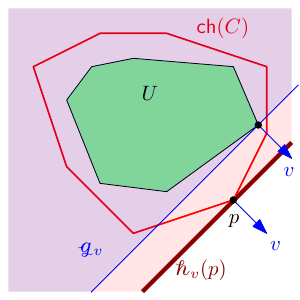}%
        \hfill%
        \hfill%
        (B) \hfill%
        \includegraphics[page=2]{figs/cover_by_intervals} \hfill%
        \phantom{}
        \captionof{figure}{}
        \figlab{cover:by:intervals}
    \end{figure}%

By \lemref{cover:reduction}, the problem of computing the smallest
subset $\dualX{\CS} \subseteq \PS$ with
$\CHX{\dualX{\CS}} \supseteq \InPoly$ reduces to computing a smallest
subset of $\ICS$ that covers $\dirs$. Therefore, using \thmref{arc-cover}, we obtain the following:

\begin{corollary}
    \corlab{inpoly}%
    Let $\PS$ be a set of $n$ points in $\reals^2$, and let
    $\InPoly \subseteq \CHX{\PS}$ be a convex polygon with $m$
    edges. A minimum set $\CS \subseteq \PS $ such that
    $\InPoly \subseteq \CHX{\PS}$ can be computed in $O(m+n \log mn)$
    time.
\end{corollary}

We first compute $\coreY{\PS}{\eps/4}$ and then compute the smallest
subset of $\PS$ that contains $\coreY{\PS}{\eps/4}$.  By applying
\corref{core-ker}, \lemref{core:algo}, and \corref{inpoly} we obtain
the following result:

\begin{theorem}
    \thmlab{eps-kernel1}%
    Let $\PS$ be a set of $n$ points in
    $\reals^2$, and let $\eps\in (0,1)$ be a parameter.  A
    $\eps$-kernel of $\PS$ of size at most $\opt_{\eps/4}$ can be
    computed in $O(n\log n)$ time.
\end{theorem}

\section{Optimal Hausdorff approximation}
\seclab{hausdorff}

In this section, we consider the following variant of the
$\eps$-kernel problem: Given a set $\PS$ of $n$ points in $\reals^2$
and a parameter $\eps$, we wish to compute a subset
$\CS \subseteq \PS$ of the smallest size such that the Hausdorff
distance between $\CHX{\CS}$ and $\CHX{\PS}$ is at most $\eps$.
Recall that the \emphi{Hausdorff distance} between two sets $X, Y
\subseteq \Re^2$ is
\begin{equation*}
    \max\Bigl(\max_{x\in X} \min_{y\in Y} \dY{x}{y},\, \max_{y\in Y}\min_{x\in X}
    \dY{x}{y}\Bigr).
\end{equation*}

\begin{figure}[h]
    \centerline{\includegraphics{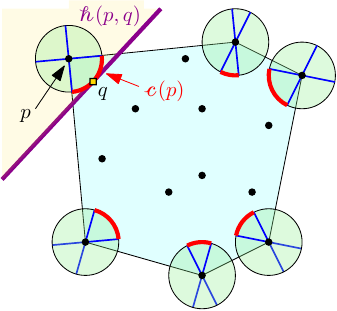}}
    \caption{$\CHX{\PS}$, circular arcs $\mathcalb{c}(\pa)$, and the
       halfplane $\hn(\pa,\pb)$. Every point on a red arc encodes a
       halfplane that must be hit. }
    \figlab{extremal}
\end{figure}

Analogous to the $\eps$-kernel problem, this problem also can be
restated as a hitting-set problem: for any direction $v$, the
halfplane $\hpDirMY{v}{\eps}$, see \defref{supporting:line}, must be
hit by a point of $\CS$.  This set of halfplanes can be interpreted as
follows: For every vertex $\pa$ of $\CHX{\PS}$, we place a unit circle
centered at the vertex of radius $\eps$. Let $\mathcalb{c}(\pa)$ be
the circular arc consisting of all the points $\pb$ on this circle
such that $\pb \pa$ intersects the interior of $\CHX{\PS}$, and that
$\pa$ is extremal in the direction $\pa - \pb$, see
\figref{extremal}. For a point $\pb \in \mathcalb{c}(\pa)$, let
$\hn(\pa,\pb)$ denote the halfplane that is (i) orthogonal to
$\pa \pb$, (ii) its boundary passes through $\pb$, and (iii) it
contains $\pa$.  The set halfplanes that have to be hit by $\CS$ is
exactly
\begin{equation*}
    \HH = \Set{ \hn(\pa,\pb) }%
    { \pa \in \VX{\CHX{\PS}}, \pb \in \mathcalb{c}(\pa)}.
\end{equation*}

\begin{figure}
    \centering
    \begin{tabular}{*{3}{c}}
      \includegraphics[page=1,width=0.2\linewidth]{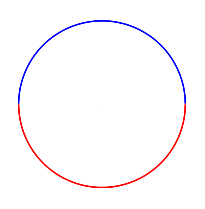}
      \includegraphics[page=2,width=0.2\linewidth]{figs/dual_circle}
      &
        \hspace*{0.25in}
      &
        \includegraphics[width=0.4\linewidth]{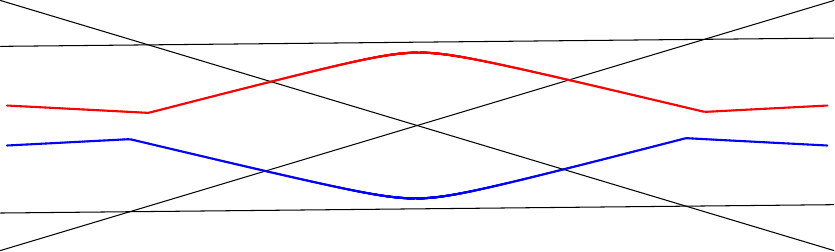}
      \\
      (a)&&(b)
    \end{tabular}
    \caption{(a) A circle and its dual hyperbolic arc. (b) In the
       dual, under the Hausdorff distance, the $\eps$-upper
       envelope and $\eps$-lower envelope are chains of hyperbolic
       arcs.}
    \figlab{dual:circle}
\end{figure}

As in \secref{optimal:i:c}, using polarity, $\HH$ is mapped to a star
shaped region $\PP$ (with respect to the origin).  $\partial \PP$
consists of the points that are polar of the lines bounding the
halfplanes in $\HH$. The problem of computing a hitting set of $\HH$
is equivalent to computing the smallest subset of $\LS = \polarX{\PS}$
that is a blocking set of $\PP$. The latter problem can be reduced to
computing the smallest arc cover in a set $\Xi$ of unit arcs. We can
follow the same approach as in \secref{optimal:i:c} and compute the
smallest-size arc cover in $O(n\opt\log n)$ time, where $\opt$ is the
size of the smallest blocking set, provided we have a data structure
for computing $\shoot{\Line}{u}$ in $O(\log n)$ time after
$O(n\log n)$ preprocessing.  Since $\partial\PP$ is now piecewise
algebraic and $\PP$ is not a polygon, the data structures
in~\cite{hs-parss-95,cg-vip-89} cannot be used directly.  Instead, we
use duality instead of polarity and segment-intersection-searching
data-structure described in~\cite{cg-fca} (which works for our setting
as well).  It is known that the dual of the lines tangent to a
circular arc is a hyperbolic arc (see for example \cite[Lemma
6.5]{chj-oagcd-21}) -- this is illustrated in
\figref{dual:circle}~(a).  The $\eps$-upper and $\eps$-lower envelopes
in the dual plane are now defined as follows:

\[ \uY{\LS}{x} = \UY{\LS}{x}-\eps\sqrt{1+x^2}\quad \mbox{and}\quad
    \lY{\LS}{x} = \LY{\LS}{x}+\eps\sqrt{1+x^2} .\] Each of them is an
$x$-monotone piecewise-hyperbolic curve. Each arc of $\uZ{\LS}$
(resp.\ $\lZ{Y}$) is a piece of the boundary of a concave (resp.\
convex) function.  Therefore $\uZ{\LS}$ can be viewed as a mountain
chain, and $\lZ{\LS}$ as a chain of craters.
See \figref{dual:circle}~(b) for an example.

The data structure in~\cite{cg-fca} can be adapted in a
straightforward manner to answer ray-shooting queries in $\uZ{\LS}$
and $\lZ{\LS}$.  In particular, we can preprocess $\uZ{\LS}$ in
$O(n\log n)$ time so that for a query ray emanating from a point lying
above $\uZ{\LS}$, its first intersection point with $\uZ{\LS}$ can be
computed in $O(\log n)$ time.  The same holds for $\lZ{\LS}$. Hence,
putting everything together, we obtain the following:

\begin{theorem}
    \thmlab{hausdorff:opt}%
    Given a set $\PS$ of $n$ points in the plane, and a parameter
    $\eps>0$, a set $\QS \subseteq \PS$ of size $\optHX{\eps}$ can be
    computed, in $O(n \optHX{\eps} \log n)$ time, such that the
    Hausdorff distance between $\CHX{\PS}$ and $\CHX{\QS}$ is at most
    $\eps$, where $\optHX{\eps}$ is the minimum size of any such set.
\end{theorem}

\section{Conclusions}

In this paper, we studied the problem of computing optimal kernels in
the plane, both in the strong and weak sense. Surprisingly, this very
natural problem had not received much attention when kernels were
developed around twenty years ago. The problem has surprisingly
non-trivial structure, and getting near linear running time to compute
them exactly required non-trivial ideas and care. A natural open
question is whether an instance-optimal $\eps$-kernel of $n$ points in
$\reals^2$ can be computed in $O(n\log n)$ time (our algorithm is
slower by a factor of $O(1/\sqrt{\eps})$, see \thmref{opt-kernel}).

\DCGVer{\bmhead{Acknowledgments} %

Work by the first author on this paper was partially supported by NSF
grants \si{IIS-18-14493} and CCF-20-07556.  Work by the second author
on this paper was partially supported by an NSF AF award CCF-1907400.

}

\RegVer{%
   \BibTexMode{%
      \DCGVer{%
         \bibliographystyle{plainurl}%
      }%
      \RegVer{%
         \bibliographystyle{alpha}%
      }%
   } }%
\BibTexMode{%
   \bibliography{opt_kernel}%
}

\BibLatexMode{\printbibliography}

\RegVer{%
\begin{figure}[h]
    \centerline{%
       \animategraphics%
       [autoplay,width=0.9\linewidth,loop]%
       {4} %
       {figs/polar_outer} %
       {1} %
       {74} %
    }
    \caption{Animation: Another example of $\CHX{\PS}$,
       $\IS(\polarX{\PS})$, and $\IS_\eps(\polarX{\PS})$.  See also
       \figref{polar-ker}.  }
    \figlab{animation}
\end{figure}
}

\end{document}